\documentclass[10pt,onecolumn,draftcls]{IEEEtran}

\def\comment#1{}
\usepackage{graphicx}

\usepackage[utf8]{inputenc} 
\usepackage[T1]{fontenc}    
\usepackage{url}            
\usepackage{booktabs}       
\usepackage{amsfonts}       
\usepackage{nicefrac}       
\usepackage{microtype}      

\usepackage{subfigure}
\usepackage{booktabs}
\usepackage{amsfonts,bbm,mathrsfs}
\usepackage{amsmath}
\usepackage{algorithm,algorithmic,bm,color}
\usepackage{graphicx,amsfonts,amsmath,amssymb,xcolor,amsthm}
\usepackage{algorithm,algorithmic,bm,color}
\usepackage{multirow}
\usepackage{wrapfig}
\usepackage{adjustbox}
\usepackage{graphics, color,  dsfont}

\usepackage{graphicx,psfrag,dsfont}
\usepackage{amsfonts,amsmath,amssymb,color}
\usepackage{bbm,setspace}

\newcommand{\argmin}{\arg\!\min}

\newcommand{\Eox}[1]{{\rm E}\left[#1\right]}

\newcommand{\LPd}[1]{\langle#1\rangle}

\newcommand{\Cmat}{{\boldsymbol C}}
\newcommand{\Dmat}{{\boldsymbol D}}

\newcommand{\Hmat}[0]{{{\boldsymbol H}}}
\newcommand{\Imat}{{\boldsymbol I}}

\newcommand{\Mmat}[0]{{{\boldsymbol M}}}

\newcommand{\Rmat}[0]{{{\boldsymbol R}}}

\newcommand{\Xmat}{{\boldsymbol X}}
\newcommand{\Ymat}[0]{{{\boldsymbol Y}}}
\newcommand{\Zmat}{{\boldsymbol Z}}

\newcommand{\ev}[0]{{\boldsymbol{e}}}
\newcommand{\fv}{\boldsymbol{f}}

\newcommand{\qv}[0]{{\boldsymbol{q}}}

\newcommand{\uv}[0]{{\boldsymbol{u}}}
\newcommand{\vv}{\boldsymbol{v}}

\newcommand{\xv}{\boldsymbol{x}}
\newcommand{\yv}{\boldsymbol{y}}

\newcommand{\zv}{\boldsymbol{z}}

\newcommand{\vvt}{\tilde{\boldsymbol{v}}}

\newcommand{\ts}{^{\top}}
\newcommand{\inv}{^{-1}}

\newcommand{\ie}{{\em i.e.}}
\newcommand{\eg}{{\em e.g.}}

\newcommand{\Uc}{\mathcal{U}}
\newcommand{\Qc}{\mathcal{Q}}
\newcommand{\Ec}{\mathcal{E}}
\newcommand{\Fc}{\mathcal{F}}
\newcommand{\Nc}{\mathcal{N}}

\newtheorem{theorem}{Theorem}

\newcommand{\mvec}[1]{{\boldsymbol #1}}

\begin{document}

\title{GAP-net for Snapshot Compressive Imaging
}

\author{Ziyi Meng, Shirin Jalali  and  Xin Yuan 
\thanks{Ziyi Meng is at Beijing University of Posts and Telecommunications, Beijing, 100876, China, mengziyi@bupt.edu.cn.
	Shirin Jalali   is with Nokia Bell Labs, 600 Mountain Avenue, Murray Hill, NJ, 07974, USA, shirin.jalali@nokia-bell-labs.com. 
	Xin Yuan is with Nokia Bell Labs, 600 Mountain Avenue, Murray Hill, NJ, 07974, USA, xyuan@bell-labs.com (corresponding author).}
}


\maketitle

\begin{abstract}
Snapshot compressive imaging (SCI) systems aim to capture high-dimensional ($\ge3$D) images in a single shot using  2D detectors. SCI devices   include two main parts: a hardware encoder and a software decoder.
The hardware encoder typically consists of an (optical) imaging system designed to capture  {compressed measurements}. The software decoder on the other hand refers to a reconstruction algorithm that retrieves the desired high-dimensional signal from those  measurements.  
In this paper, using deep unfolding ideas, we propose an SCI recovery algorithm, namely GAP-net, which unfolds the  generalized alternating projection (GAP) algorithm.  At each stage, GAP-net passes  its current estimate of the desired signal through  a trained convolutional neural network (CNN). The CNN  operates as a denoiser that projects the estimate back to the desired  signal space. For the GAP-net that employs  trained auto-encoder-based denoisers, we  prove a probabilistic global convergence result. 
Finally, we investigate the performance  of GAP-net in solving video SCI and spectral SCI problems. In both cases, GAP-net demonstrates competitive performance on both synthetic and real data. In addition to having high accuracy and high speed, we show that GAP-net is flexible with respect to signal modulation implying that  a trained GAP-net decoder can be applied in different systems. Our code is at \url{https://github.com/mengziyi64/ADMM-net}.
\begin{IEEEkeywords}
Compressive imaging, Compressive sensing, Deep learning, Generative alternating projection, Snapshot, Convolution neural network, Convergence, Denoising
\end{IEEEkeywords}
\end{abstract}

\section{Introduction}
\label{intro}
Recent advances in artificial intelligence and robotics have resulted in an unprecedented demand for computationally-efficient high-dimensional (HD) data capture and processing devices. However, existing optical sensors usually can only directly capture no more than two-dimensional (2D) signals. Capturing 3D or higher dimensional signals  remains challenging since sensors that directly perform 3D data acquisition do not yet exist. 

In recent years, snapshot compressive imaging (SCI) systems that employ 2D detectors to capture HD ($\ge$3D) signals have been shown to be a promising solution to address this challenge. 
Different from conventional cameras, SCI systems perform {sampling} on a set of consecutive images---video frames (\eg, CACTI~\cite{Patrick13OE}) or spectral channels (\eg, CASSI~\cite{Wagadarikar08CASSI})---in accordance with the sensing matrix and {integrate} these sampled signals along time  or spectrum to obtain the final compressed measurements. Using this technique, SCI systems~\cite{Gehm07,Hitomi11ICCV,Reddy11CVPR} can capture high-speed motion~\cite{Yuan14CVPR,Yuan16BOE} or high-resolution spectral information~\cite{Wagadarikar09CASSI,Yuan15JSTSP}, with low memory, low bandwidth, low power and potentially low cost. 

There are two main ingredients in an SCI system: a hardware encoder and a software decoder.
The hardware encoder is typically  an (optical) imaging system that is designed to capture  {\em compressed measurements} of the desired signal and the software decoder refers to the algorithm that recovers the desired HD data from those measurements. 
The underlying principle of the {\em hardware encoder} is to modulate the HD data using a speed higher than the capture rate of the camera. To achieve this goal, typically,  a coded aperture (binary mask) or some other spatial light modulators is  employed. 
In this paper, we focus on the {\em software decoder}, \ie, the reconstruction algorithm, which, until recently, had precluded  wide application of SCI due to the slow speed of existing reconstruction methods. Thanks to deep learning, various SCI recovery algorithms based on convolutional neural networks (CNNs) have been developed in recent years~\cite{Qiao2020_APLP,Wang19TIP,Meng20ECCV_TSAnet,Cheng20ECCV_Birnat}. These algorithms  significantly  speed up the reconstruction process compared to the previous methods. This has paved the way to daily applications of SCI systems. 

Reconstruction algorithms for SCI can  be divided into four classes as follows: (Refer to  Fig.~\ref{fig:methods}.)
\begin{itemize}
	\item [(a)] Iterative optimization-based algorithms that are designed  assuming that the source structure is captured by   (typically simple)  priors (or  regularization functions), such as sparsity~\cite{Yang20_TIP_SeSCI} or total variation (TV) ~\cite{Liu19_PAMI_DeSCI,Yuan16ICIP_GAP};
	\item [(b)] End-to-end CNNs (E2E-CNNs) that are trained with a large amount of training data consisting of measurements and (optionally) masks as  input and the ground truth \ie,  the underlying desired signal, as  output; during testing,  feeding the captured measurements into the trained E2E-CNN, it is expected to produce the desired signal  instantaneously~\cite{Miao19ICCV,Meng2020_OL_SHEM,Cheng20ECCV_Birnat};
	\item [(c)] Deep unfolding algorithms, where, instead of one CNN, a sequence of $K$ CNNs are trained to map the measurements to the desired signal~\cite{Ma19ICCV,Wang19_CVPR_HSSP}; unlike E2E-CNNs,  here, each stage consists of two parts: i) linear transformation and ii) passing the signal through a CNN operating as a denoiser, and
	\item [(d)] Plug-and-play (PnP) methods~\cite{Sreehari16PnP,Venkatakrishnan_13PnP} in which a pre-trained denoising network is plugged into one of the iterative algorithms of class (a)~\cite{Qiao2020_APLP,Yuan2020_CVPR_PnP}.  
\end{itemize}

\begin{figure}
	\centering
	\includegraphics[width=1\linewidth]{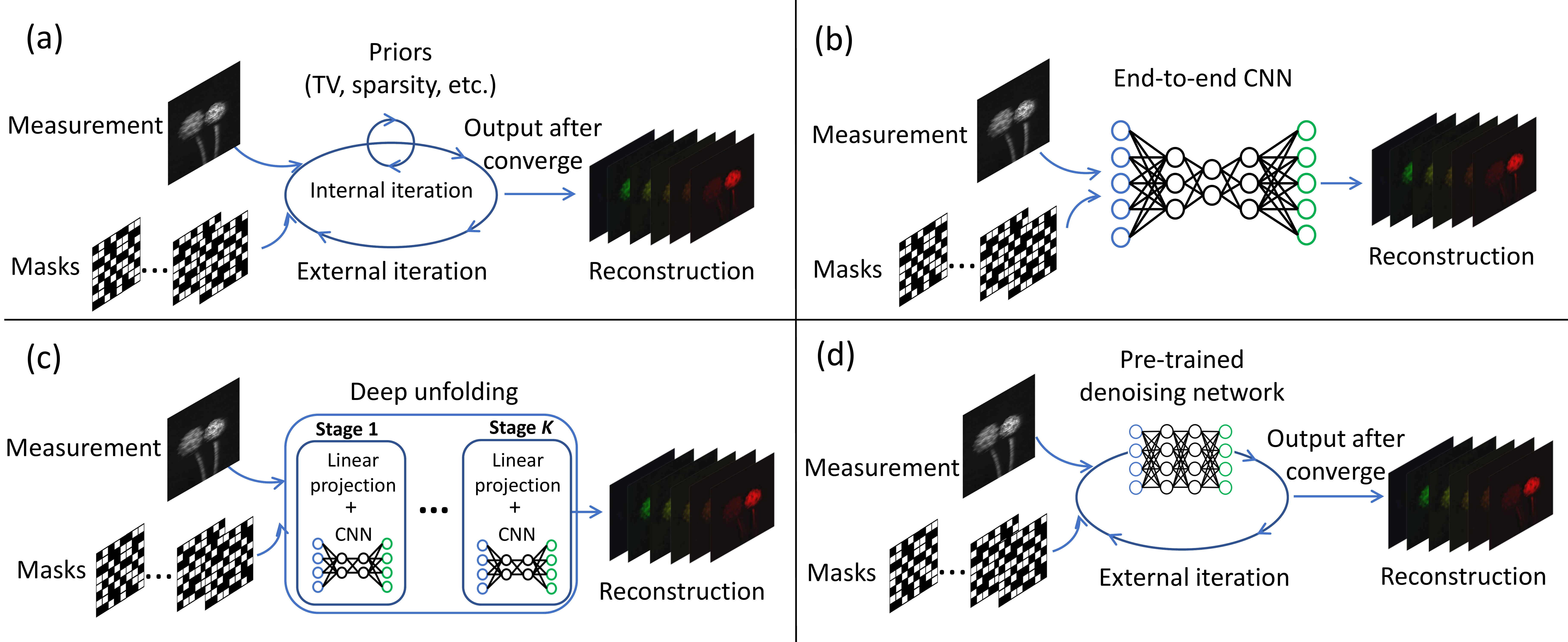}
	\caption{\small Four different frameworks for SCI reconstruction: (a) conventional iterative optimization algorithms using priors, (b) end-to-end CNN with the measurement and masks as the input and the desired signal as output, (c) deep unfolding approach using $K$  CNNs to reconstruct the datacube, and (d) Plug-and-play framework integrating pre-trained denoising networks into iterative algorithms to reconstruct the desired 3D cube.}
	\label{fig:methods}
\end{figure}

Iterative optimization-based methods have been studied extensively in the literature   for a diverse set of priors. 
The speed and quality of such algorithms are variant, with faster versions (\eg, those using TV~\cite{Bioucas-Dias2007TwIST,Yuan16ICIP_GAP}) typically not achieving high reconstruction quality, and high performance versions   usually being very slow~\cite{Liu19_PAMI_DeSCI}. However, even a relatively  fast algorithm of class (a) is currently considerably slower than CNN-based algorithms of classes (b)-(d)~\cite{Qiao2020_APLP}. 
E2E-CNN algorithms require a large training data set, a long training time and  large GPU memory. In some applications this is infeasible, for instance when dealing with large-scale data sets~\cite{Yuan2020_CVPR_PnP}. PnP algorithms are flexible and can be used in large-scale data sets, since they do not need to retrain the network. However, their performance is mainly limited by the pre-trained denoising network. For instance,  efficient and flexible  denoising algorithms for hyperspectral images are not yet available.

\subsection{Related work}
Various   optimization-based SCI algorithms, such as TwIST~\cite{Bioucas-Dias2007TwIST}, GAP-TV~\cite{Yuan16ICIP_GAP}, GMM~\cite{Yang14GMM,Yang14GMMonline} and DeSCI~\cite{Liu19_PAMI_DeSCI}, using different priors have been developed. The state-of-the-art algorithm, DeSCI~\cite{Liu19_PAMI_DeSCI}, applies  weighted nuclear norm minimization~\cite{Gu14CVPR} of non-local similar patches in  video frames into the alternating direction method of multipliers (ADMM) framework~\cite{Boyd11ADMM}. While DeSCI can recover clear images, it takes  about two hours for it to recover 8 frames (each of size $256 \times 256$), which precludes its wide application. Inspired by the recent successes  of deep-learning-based solutions on image restoration~\cite{NIPS2012_4686,Zhang17TIP}, in recent years, researchers have explored  application of deep learning in computational imaging as well~\cite{Chang17ICCV,Iliadis18DSPvideoCS,Jin17TIP,Kulkarni2016CVPR,LearningInvert2017,George17lensless}. 
Deep unfolding (or unrolling)~\cite{hershey2014deep_unfold} has been used for compressive sensing (CS) and various solvers such as ADMM-net~\cite{Yang_NIP16_ADMM-net}, ISTA-net~\cite{zhang2018ista} and learned D-AMP~\cite{Metzler_Learned_DAMP_Nip17} have been proposed.\footnote{It has been observed in several papers~\cite{Jalali19TIT_SCI} and experiments~\cite{Liu19_PAMI_DeSCI} that AMP does not converge well in SCI applications, due to the special structure of sensing matrix in SCI and ADMM usually outperforms ISTA.}
Regarding the E2E-CNN, the recurrent neural network (RNN) proposed in~\cite{Cheng20ECCV_Birnat} can now provide competitive results as DeSCI for video SCI and the best spectral SCI result is obtained by the spectial-spectral model proposed in~\cite{Meng20ECCV_TSAnet}. 

Deep unfolding has also been used for video SCI in~\cite{Li2020ICCP,Han_AAAI2020_FISTA,Ma19ICCV} and spectral SCI~\cite{Wang19_CVPR_HSSP}. However, these methods are inspired by sparse coding in each stage and the resulting reconstruction quality of real data is low.  Furthermore, no convergence guarantees have been established for such algorithms yet. 
The most recent paper~\cite{Li2020ICCP} employs a similar idea of using a denoiser at each stage. However, the paper  only explores video SCI and provides  no theoretical analysis of the proposed method. 

Note that existing models were developed either for video SCI or spectral SCI. For instance, it is reasonable to use RNNs in video SCI, but it is not intuitive to apply RNNs to spectral SCI. 
In this paper, we aim to develop a unified model that works efficiently for both video SCI and spectral SCI.

\subsection{Contributions of this work}
Motivated by the pros and cons of different approaches, in this paper, we propose a deep unfolding SCI recovery algorithm called GAP-net. Here are some key properties of GAP-net that are described in details later:
\begin{itemize}
	\item GAP-net is a deep unfolding algorithm that employs generalized alternating projection (GAP)~\cite{Liao14GAP} framework. GAP is known to be a more efficient framework than  ADMM \cite{Yuan2020_CVPR_PnP} for SCI reconstruction. 
	\item GAP-net, a recovery algorithm for  both video and spectral SCI, treats the trained CNN at each stage of the decoder as  a {\em denoising} network.  This is in the same spirit as learned D-AMP \cite{Metzler_Learned_DAMP_Nip17}, a recovery method for image CS and \cite{Li2020ICCP}  for video SCI.  There exist other SCI recovery algorithms in the literature that are based on deep unfolding. In such methods, the  CNN in each stage of the unfolding includes the transformation (implemented by convolutional or fully-connected neural  networks) and shrinkage (implemented by ReLU, etc.). 
	\item Based on the above denoising perspective,  GAP-net is related to the PnP framework for SCI~\cite{Yuan2020_CVPR_PnP}. We  theoretically prove and experimentally verify its convergence to the desired result\footnote{We observed some error in the proof of the global convergence of PnP-GAP in~\cite{Yuan2020_CVPR_PnP}. Specifically, the lower bound of the second term in Eq. (25) in~\cite{Yuan2020_CVPR_PnP} should be 0. Therefore, the proof of global convergence for PnP-GAP presented in \cite{Yuan2020_CVPR_PnP}  does not hold. The authors have updated the manuscript on arXiv with a local convergence result. In this paper, using  concentration of measure tools, we prove  global convergence of GAP-net.}.
	\item GAP-net achieves state-of-the-art (in some cases, competitive) performance in our experimental results done both with synthetic data and  real data, in both video SCI and spectral SCI. 
\end{itemize}

\section{Video SCI and Spectral SCI}
Fig.~\ref{fig:principle} depicts the schematic diagrams of both video (left) and spectral (right) SCI. As mentioned earlier, the underlying principle in both of them is to modulate the datacube using a speed higher than the capture rate of the camera.
In video SCI, different modulation  solutions, such as a shifting mask~\cite{Patrick13OE} and different digital micromirror device (DMD) patterns~\cite{Hitomi11ICCV}, both of which  are active modulation methods, have been proposed. On the other hand, in spectral SCI, the modulation can be implemented by a fixed mask plus a disperser~\cite{Wagadarikar08CASSI}, which is a passive modulation approach and thus more power efficient. The key idea is that, using a disperser,  a fixed modulation pattern can be shifted to be different for different wavelengths.

\begin{figure}[tbp!]
	\centering
	\includegraphics[width=1\linewidth]{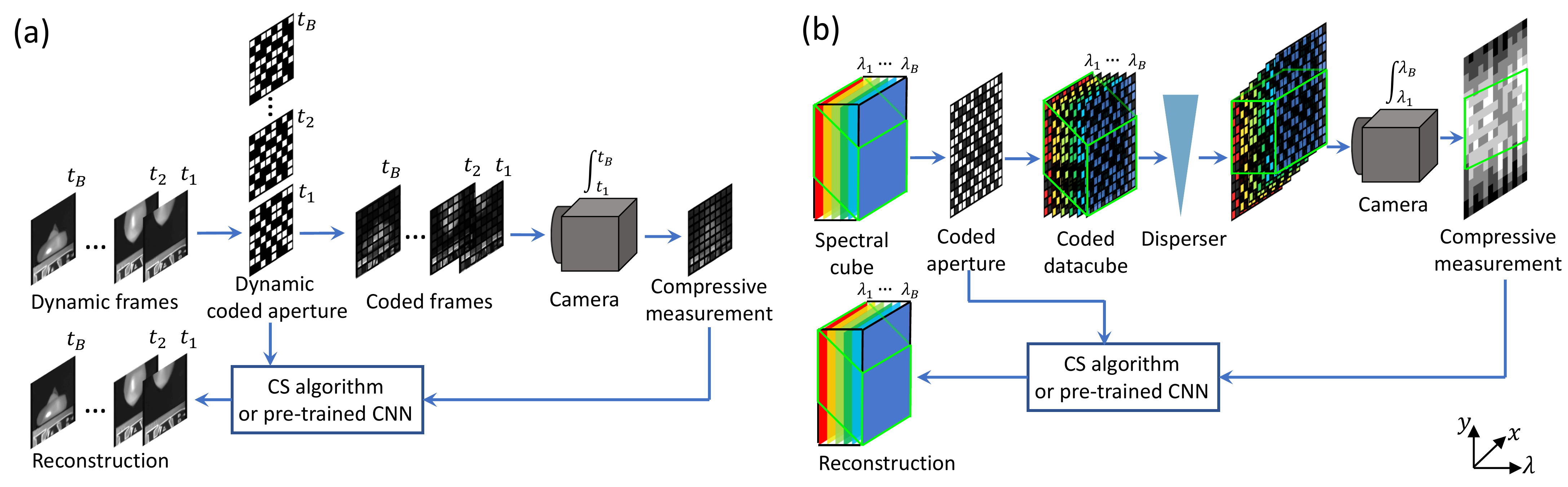}
	\caption{\small Schematic diagrams of (a) video SCI and (b) spectral SCI. In video SCI, a sequence of masks are imposed on the high-speed scene during one integration time. The captured 2D measurement thus encompasses the information of the desired spatial-temporal datacube, which can be reconstructed by the CS or deep learning algorithms. In spectral SCI, the spatial-spectral datacube is first modulated by a fixed physical mask and then the modulated datacube is sheared by a disperser. The 2D coded measurement thus includes the information of the spectral datacube, which is the desired 3D signal. }
	\label{fig:principle}
\end{figure}

\subsection{Forward Model of Video SCI~\label{Sec:SCImodel}}
Consider a video SCI encoder, where  a  $B$-frame video $\Xmat \in \mathbb{R}^{n_x \times n_y \times B}$  is modulated  and mapped  by $B$ sensing matrices (masks) $\Cmat\in \mathbb{R}^{n_x \times n_y \times B}$  to a single measurement frame  $\Ymat \in \mathbb{R}^{n_x\times n_y}$ as
\begin{equation}\label{Eq:System}
  \Ymat = \sum_{b=1}^B \Cmat_b\odot \Xmat_b + \Zmat,
\end{equation}
where $\Zmat \in \mathbb{R}^{n_x \times n_y }$ denotes the noise; $\Cmat_b = \Cmat(:,:,b)$ and $\Xmat_b = \Xmat(:,:,b) \in \mathbb{R}^{n_x \times n_y}$ denote the $b$-th sensing matrix (mask) and the corresponding video frame, respectively; $\odot$ denotes the Hadamard (element-wise) product. 
Using vectoring operator, define $\yv = \text{Vec}(\Ymat) \in \mathbb{R}^{n_x n_y}$ and $\zv= \text{Vec}(\Zmat) \in \mathbb{R}^{n_x n_y}$. Similarly,  define  $\xv \in \mathbb{R}^{n_x n_y B}$ as
\begin{equation}
\xv = \text{Vec}(\Xmat) = [\text{Vec}(\Xmat_1)\ts,..., \text{Vec}(\Xmat_B)\ts]\ts.
\end{equation}
Using these definitions, the measurement process defined in \eqref{Eq:System} can be expressed  as 
\begin{equation}\label{Eq:ghf}
\yv = \Hmat \xv + \zv.
\end{equation}
Unlike standard CS, the sensing matrix $\Hmat \in \mathbb{R}^{n_x n_y \times n_x n_y B}$ in video SCI is  highly structured  and sparse. More precisely, it can be written as  the  concatenation of $B$  diagonal matrices as
\begin{equation}\label{Eq:Hmat_strucutre}
\Hmat = [\Dmat_1,...,\Dmat_B],
\end{equation}
where, for  $b =1,\dots B$, $\Dmat_b = \text{diag}(\text{Vec}(\Cmat_b)) \in {\mathbb R}^{n \times n}$ with $n = n_x n_y$.
Here,  the {\em sampling rate} is equal to $1/B$. It has been proven in~\cite{Jalali19TIT_SCI} that, if the signal is structured enough,  there exist SCI recovery algorithms with bounded reconstruction error, even for $B>1$.

\subsection{Forward Model of Spectral SCI}
Recall Fig. \ref{fig:principle}(b) which shows a schematic diagram of an {\em spectral} SCI system. Let $\Xmat^0 \in {\mathbb R}^{n_x \times n_y \times n_\lambda}$ and $\Mmat^0\in {\mathbb R}^{n_x \times n_y}$ denote the original spectral cube and the fixed mask modulation, respectively. 
After  scene $\Xmat^0$ passes through the mask, 
let ${\Xmat}' \in {\mathbb R}^{n_x \times n_y \times n_\lambda}$  denote its modulated version, in which images at different wavelengths are modulated separately, \ie, 
for $b = 1,\dots, n_{\lambda}$, we have
\begin{equation}
{\Xmat}' (:,:,b) = {\Xmat}^0 (:,:,b) \odot \Mmat^0,
\end{equation} 
where $\odot$ represents the element-wise multiplication and ${\Xmat}^0 (:,:,b)$ denotes the $b$-th frame in the spectral cube of $\Xmat^0$.

After passing the disperser, the cube $\Xmat'$ is tilted and is considered to be sheared along the $y$-axis. 
Let $\Xmat''\in {\mathbb R}^{N_x \times (Ny + N_{\lambda}-1) \times N_{\lambda}}$ denote the tilted cube 
and assume that $\lambda_c$ is the reference wavelength, \ie, image $\Xmat'(:,:,n_{\lambda_c})$ is not sheared along the $y$-axis,
%
and therefore
\begin{equation}
\Xmat'' (u,v, b) = \Xmat'(x, y + d(\lambda_b - \lambda_c), b),
\end{equation}
where $(u,v)$ indicates the coordinate system on the detector plane, $\lambda_b$ is the wavelength at $b$-th channel and $\lambda_c$ denotes the center-wavelength. Then, $d(\lambda_b -\lambda_c)$ signifies the spatial shifting  for the $b$-th channel.

Considering that the detector sensor integrates all the light in the wavelength $[\lambda_{\rm min}, \lambda_{\rm max}]$, the compressed measurement at the detector $y(u,v)$ can thus be written as 
\begin{equation}
{ y(u,v) = \int_{\lambda_{\rm min}}^{\lambda_{\rm max}} x'' (u,v, b_{\lambda}) d \lambda,}
\end{equation}
where $x''$ denotes the analog (continuous) representation of $\Xmat''$.

In the discretized version, the captured 2D measurement $\Ymat \in {\mathbb R}^{n_x \times (ny + n_{\lambda}-1)}$ can be written as 
\begin{equation}
{ \Ymat = \sum_{b=1}^{n_{\lambda}}  \Xmat'' (:,:, n_{\lambda}) + \Zmat}.
\end{equation}
In other words, $\Ymat$ is a {\em compressed} frame which is formed by a function of the desired information corrupted by measurement noise $\Zmat\in {\mathbb R}^{n_x \times (n_y + n_{\lambda}-1)}$. 

For the convenience of model description, we further let 
$\Cmat\in \mathbb{R}^{n_x \times (n_y + n_{\lambda}-1) \times n_{\lambda}}$ denote the {\em shifted} version of the mask $\Mmat^0$ corresponding to different wavelengths, \ie,
\begin{equation}
\Cmat (u,v,b) = \Mmat^0(x, y + d(\lambda_b - \lambda_c)).
\end{equation}
Similarly, for each signal frame at a different wavelength, the shifted version is $\Xmat \in {\mathbb R}^{n_x \times (n_y + n_{\lambda}-1) \times n_{\lambda}}$,
\begin{eqnarray}
\Xmat (u,v, b) = \Xmat^0(x, y + d(\lambda_b - \lambda_c), b).
\end{eqnarray}
Following this, the measurement $\Ymat$ can be represented as
\begin{equation}
{ \Ymat = \sum_{b=1}^{n_{\lambda}}   \Xmat (:,:, b)  \odot \Cmat (:,:, b) + \Zmat}. \label{Eq:sensing_Matrix}
\end{equation}

\noindent \textbf{Vectorized Formulation.}
Let ${\rm vec}(\cdot)$  denote the matrix vectorization operation, \ie, concatenating columns of a matrix into one vector. 
Then, we define
$\yv = {\rm vec}(\Ymat)$, 
$\zv = {\rm vec}(\Zmat) \in {{\mathbb R}^{n_x(n_y+n_\lambda-1)}}$ and
\begin{equation}
{\xv} = \left[\begin{array}{c}
{\xv}^{(1)}\\
\vdots\\
{\xv}^{(n_{\lambda})}
\end{array}\right] \in {{\mathbb R}^{n_x (n_y + n_{\lambda}-1)n_{\lambda}}} 
\end{equation}
where, for 
$b=1,\dots, n_\lambda$, $\xv^{(b)} = {\rm vec}( \Xmat (:,:, b))$.

In addition, we define the sensing matrix as
\begin{equation}
\Hmat = \left[\Dmat_1, \dots, \Dmat_{n_{\lambda}}\right] \in   {{\mathbb R}^{n \times  n_{\lambda} n}}, \label{Eq:Phimat}
\end{equation}
where $n = n_x (n_y + n_{\lambda}-1)$ and 
$\Dmat_{b} = {\rm Diag} ({\rm vec}(\Cmat (:,:, b)))$
is a diagonal matrix with ${\rm vec}(\Cmat (:,:, b))$ as its diagonal elements.
As such, we then can rewrite the matrix formulation of \eqref{Eq:sensing_Matrix} as
\begin{equation}
\yv = \Hmat {\xv} + \zv. \label{Eq:CS_forwared}
\end{equation}
This shares the same format as video SCI. But the only difference is that now the measurement is $\Ymat\in {\mathbb R}^{n_x \times (n_y +n_{\lambda} -1)}$ and after we recover  $\Xmat$, the desired spatial-spectral cube is obtained by shifting each channel of $\Xmat$ back into its original position to get $\Xmat^0$.


In the following, we use the unified formulation described in~\eqref{Eq:ghf} to develop a decoding algorithm that, given $\yv$ and $\Hmat$,   recovers $\xv$.

\section{GAP-net for SCI}

An SCI reconstruction algorithm aims at solving the following optimization: 
\begin{equation}
  \hat{\xv} = \argmin_{\xv} \frac{1}{2}\|\yv- \Hmat \xv\|_2^2  + \tau \Omega(\xv), \label{Eq:forward}
\end{equation} 
where $\Omega(\xv)$ is a regularization term that is designed to capture the source structure.    To solve \eqref{Eq:forward}, we propose GAP-net, an algorithm that takes advantage of both deep unfolding and generalized alternating projection ideas. GAP-net is also inspired by the PnP framework~\cite{Chan2017PlugandPlayAF}, where unlike PnP,  a  different  denoising network is  trained in each stage of the algorithm. In fact, GAP-net integrates  a sequence of trained  denoising  networks  into the GAP framework to provide high-quality reconstruction as well as high flexibility. 

Specifically, GAP-net involves  $K$  stages shown in Fig.~\ref{fig:GAP_net}.  At stage $k$, $k=1,\ldots,K$, let $\vv^{(k)}$ denote our current estimate of the desired signal. Also, let  $\xv^{(k)}$ denote an auxiliary vector of the same dimension as $\vv^{(k)}$. GAP-net updates $\vv^{(k)}$ and $\xv^{(k)}$  as follows:
\begin{itemize}
	\item Updating  $\xv$: $\xv^{(k)}$ is updated via a Euclidean projection of
	$\vv^{(k)}$ on the linear manifold ${\cal M}: \yv = \Hmat \xv$. That is,
	\begin{equation}
	  \xv^{(k+1)} =  \vv^{(k)} + \Hmat\ts (\Hmat \Hmat\ts)\inv (\yv - \Hmat \vv^{(k)}). \label{Eq:x_k+1}
	\end{equation}
	Recalling~\eqref{Eq:Hmat_strucutre}, note that $\Rmat \stackrel{\rm def}{=} \Hmat\Hmat\ts$ is a diagonal matrix and thus the inversion operation is straightforward. This has been observed and used before in a number of SCI algorithms~\cite{Liu19_PAMI_DeSCI,Yuan16ICIP_GAP}. 
	\item Updating $\vv$: After the projection, the goal of the next step is to bring $\xv^{(k+1)} $  closer  to the  desired signal domain. In GAP-net this is achieved  by employing an appropriate trained denoiser ${\cal D}_{k+1}$ and letting
	\begin{equation}
	{  \vv^{(k+1)} = {\cal D}_{k+1}(\xv^{(k+1)}).} \label{Eq:Denoise_GAP}
	\end{equation}
\end{itemize}  

GAP-net and the re-purposed ADMM-net (Fig.~\ref{fig:ADMM_net}) proposed in this paper are different from the original ADMM-net~\cite{Yang_NIP16_ADMM-net}, the deep-tensor-ADMM-net~\cite{Ma19ICCV} and tensor-FISTA-net~\cite{Han_AAAI2020_FISTA} that  aim to mimic the sparse coding in each stage, \ie, including the transformation (learned by neural networks and usually a fully connected or convolutional network) and shrinkage (implemented by ReLU, etc.). In the following we compare GAP-net 
with ADMM-net. 

\begin{figure}[tbp!]
	\centering
	\includegraphics[width=1\linewidth]{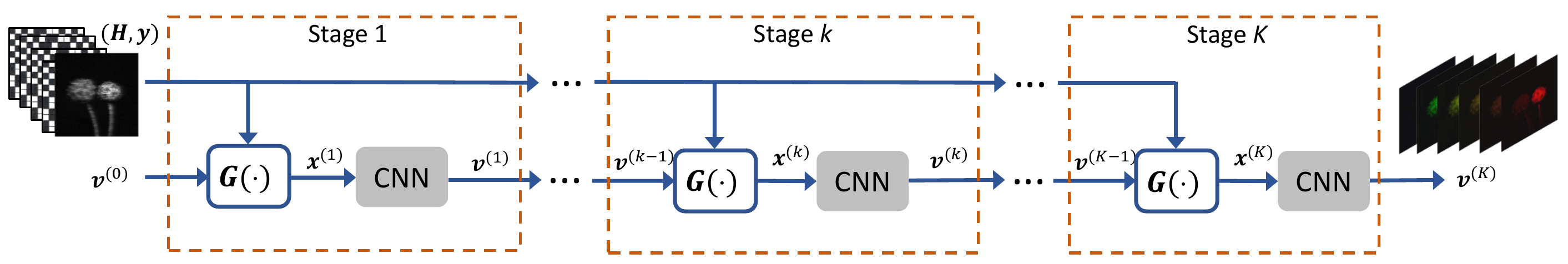}\\
	\caption{\small GAP-net with $K$ stages, where each stage is composed of a projection $G(\cdot)$ representing the operation in \eqref{Eq:x_k+1} and a CNN playing the role of denoising. $\vv^{(0)} = \Hmat\ts\yv$.}
	\label{fig:GAP_net}
\end{figure}

\begin{figure}[htbp!]
	\centering
	\includegraphics[width=1\linewidth]{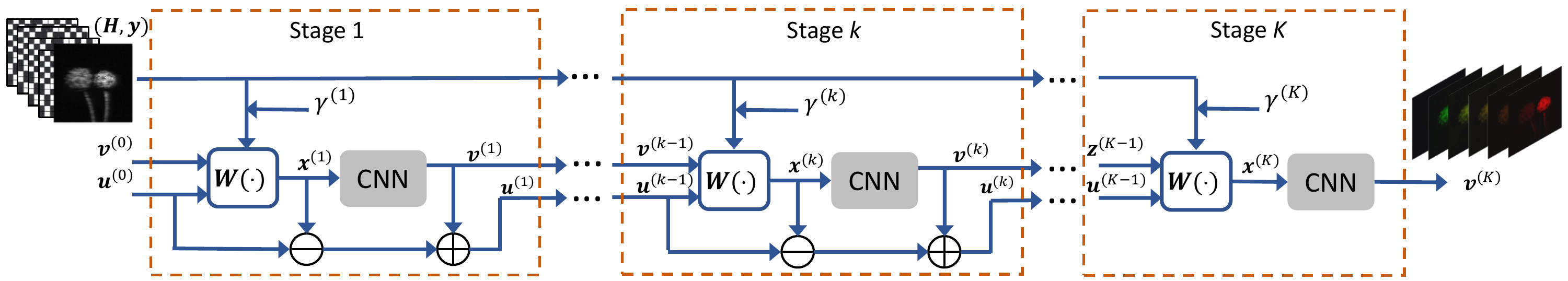}\\
	\caption{\small ADMM-net with $K$ stages, where each stage is composed of a projection $W(\cdot)$ representing the operation in \eqref{Eq:admm_x_k+1} and a CNN playing the role of denoising. $\vv^{(0)} = \Hmat\ts\yv$.}
	\label{fig:ADMM_net}
\end{figure}

\subsection{GAP-net vs. ADMM-net}
Without deriving details (refer to~\cite{Qiao2020_APLP}), we list the three steps that are  in each stage of ADMM-net:
\begin{itemize}
	\item Solve $\xv$ by $\xv^{(k+1)} = \argmin_{\xv} \frac{1}{2}\|\yv- \Hmat \xv\|_2^2 + \frac{\tau}{2}\|\xv- (\zv^{(k)}+\uv^{(k)})\|_2^2$. This has a closed form solution
	\begin{equation}
	{\xv}^{(k+1)}  = \left[\Hmat\ts\Hmat + \gamma \Imat \right]\inv\left[\Hmat\ts\yv + (\vv^{(k)} + \uv^{(k)})\right].
	\label{Eq:admm_x_k+1} 
	\end{equation}
	Due to the special structure of $\Hmat$, this can be solved in one shot~\cite{Liu19_PAMI_DeSCI}.
	\item Solve $\vv$ by ${\vv}^{(k+1)} = \argmin_{\vv} [\frac{\gamma}{2}\|\vv - (\xv^{(k+1)} - \uv^{(k)})\|_2^2 + \tau \Omega(\vv)]$. As in GAP-net, we solve this by CNN denoising
	\begin{equation}
	{\vv^{(k+1)} = {\cal D}_{k+1}(\xv^{(k+1)}- \uv^{(k)})}. \label{Eq:CNN_GAP}
	\end{equation}
	\item Update the auxiliary variable $\uv$ by
	$\uv^{(k+1)} = \uv^{(k)} -({\xv}^{(k+1)} - \vv^{(k+1)})$.
\end{itemize}
The flow-chart of ADMM-net is shown in Fig.~\ref{fig:ADMM_net}. 

Comparing with Fig.~\ref{fig:GAP_net}, clearly, GAP-net has fewer parameters and operations and thus it will be more efficient than ADMM-net. We also verify this later in our experiments by comparing the running times of different methods.

\section{Theoretical Analysis}
Let ${\cal Q}$ denote  a compact subset of ${\mathbb R}^{nB}$  representing the class of signals we are interested in. For instance, in  SCI of $B$-frame videos,  $\cal Q$ is defined as the set of all $B$-frame natural videos. GAP-net is an SCI reconstruction algorithm designed to recover  $\xv\in{\cal Q}$ from  measurements $\yv=\Hmat\xv+\zv$, where $\Hmat$ is defined in  \eqref{Eq:Hmat_strucutre}.  

GAP-net is a $K$-stage algorithm that employs a potentially different  trained CNN at each stage. The role of these CNNs is to  map the estimate obtained at stage $k$  as $ \xv^{(k+1)} =  \vv^{(k)} + \Hmat\ts (\Hmat \Hmat\ts)\inv (\yv - \Hmat \vv^{(k)})$ back to the signal space $\cal Q$.  
In this section, we prove  that using auto-encoder-based (or generative-function-based) denoisers, for $B$ small enough relative  to i) the level of structuredness of signals in $\cal Q$ and ii) parameters of the auto encoders (that in turn depend on how efficiently they  take advantage of the source structure), with high probability, GAP-net converges to the vicinity of the desired input signal\footnote{For PnP-type algorithms, various conditions on the denoiser, such as the contraction denoiser~\cite{PnP2019ICML}, Lipschitz continuous~\cite{Mertzler14Denoising}, bounded denoiser~\cite{Chan2017PlugandPlayAF} and  non-expansive denoiser~\cite{Metzler_Learned_DAMP_Nip17}, have been studied in the literature.}.

We consider using auto-encoders (AEs)~\cite{Pascal_AE_denosing} to play the role of denoisers  in GAP-net. An AE for set $\cal Q$ can be characterized by an encoder $e: \mathds{R}^{nB}\to {\cal U}^{\eta}$ and a decoder (or generative function) $g:{\cal U}^{\eta}\to\mathds{R}^{nB}$.  The generative function $g: {\cal U}^{\eta} \rightarrow {\cal X}^{nB}$ is said to cover set $\cal Q$ with distortion $\delta$, if
\begin{equation}
	\sup_{\xv \in {\cal Q}} \min_{\fv \in {\cal U}^{\eta}}\frac{1}{\sqrt{nB}}\|g(\fv) -\xv\|_2 \le \delta.
\end{equation}
A generative-function-based GAP-net, at stage $k$, employs generative function $g_k$ to perform the denoising (or projection to the signal domain) operation as follows:
\begin{equation}
{\vv}^{(k)} = g_k(\fv^{(k)}), ~~{\text{with}}~~   \fv^{(k)} = \argmin_{\fv \in {\mathbb R}^{\eta_k}}\|\xv^{(k)} -g_k(\fv)\|_2.\label{eq:gap-net-gf-based}
\end{equation}
In the case of an AE-based GAP-net, the encoder of the AE is essentially trained to perform the minimization operation in \eqref{eq:gap-net-gf-based}.

In the following we prove a global convergence result for such an AE-based GAP-net. In the theorem we assume that $g_k$ is $L_k$ Lipschitz. Also, at stage $k$, we define
\begin{equation}
 \tilde{\vv}_k = g_k(\tilde{\fv}^{(k)}) ~~{\text{with}}~~ \tilde{\fv}^{(k)} = \argmin_{\tilde{\fv} \in {\mathbb R}^{\eta_k}}\|{\xv}^* -g_k(\tilde{\fv})\|_2.
\end{equation}
In other words, $\tilde{\vv}_k$ denotes the best estimate (closest to the ground truth $\xv^*$) that $g_k$ can provide.
Finally, we assume that  the elements of the mask (\ie, sensing matrix) are i.i.d. Gaussian, \ie, $D_{b,i} \stackrel{i.i.d.}{\sim} {\cal N}(0,1)$.  We also assume that every $\xv\in{\cal Q}$ is bounded such that $\|\xv\|_\infty \le \rho/2$. Let ${\cal X}\stackrel{\rm def}{=}[-\rho/2,\rho/2]$ denote the source alphabet.
\begin{theorem} \label{The:GAP_SCI_pro_2}
	 Consider the sensing model of SCI. Assume that generative function $g_k: {\cal U}^{\eta_k} \rightarrow {\cal X}^{nB}$ covers set ${\cal Q}$ with distortion $\delta_k$ and ${\cal U}\subset[-1,1]$. Further assume that $\delta_K\le \dots \le \delta_1$ and $\eta_1\le \dots \le \eta_K$. Choose free parameter $\zeta\in(0,1)$ and for $k=1,\ldots,K-1$,  define
	\begin{equation}
	    \gamma_k=  L_k  \delta_k^{\zeta}\sqrt{\eta_k \over nB},
	\end{equation}
	and 
	\begin{equation}
	\alpha_k= 2(1+2{B})(\gamma_k(1+{1\over 1-\gamma_k}) )^{1\over 2}.
	\end{equation}
	Then, given $\lambda\in (0, 0.5-\alpha_k)$, if ${1 \over \sqrt{nB}}\|\vv^{(k)}-\vvt_k\|\geq \delta_k$ and ${1 \over \sqrt{nB}}\|\vv^{(k-1)}-\vvt_k\|\geq \delta_k$, we have
	\begin{equation}
	 	{1\over \sqrt{nB}}\|\vvt_k- \vv^{(k)}\|_2
	\leq {2\over \sqrt{nB}}(\lambda + \alpha_k) \|\vvt_{k-1}- \vv^{(k-1)}\|_2+(\lambda + \alpha_k)(\delta_k+\delta_{k-1})+2{B}\delta_k.
	\end{equation}	
	with a probability larger than 
	\begin{equation*}
	  1-\sum_{k=1}^{K-1}\exp\left[-{2\lambda^2 n\delta_k^4(1-\gamma_k)^4\over  4 {B}^2\rho^4 }+ 2\ln 2((1-\zeta) \log {1\over \delta_k }+1)\eta\right].
	\end{equation*}
\end{theorem}
\begin{proof} 
	The full proof is shown in Section~\ref{Appex:1}.
\end{proof}

\section{Network Structure and Training Details}
As explained earlier, GAP-net consists of $K$ stages. At each stage a trained CNN is used as a denoiser. While various CNN structures  have been studied in the literature for denoising, in our experiments, we  test the following structures: auto-encoder~\cite{Pascal_AE_denosing}, U-net~\cite{RFB15a}, ResNet~\cite{He_ResNet_cvpr16} and DnCNN~\cite{Zhang17TIP}. 

\subsection{Denoising Network Structure}
In our experiments, we implemented GAP-net using four different denoisers, including a 14-layer auto-encoder~\cite{Pascal_AE_denosing}, a 17-layer DnCNN~\cite{Zhang17TIP}, a 10-layer ResNet~\cite{He_ResNet_cvpr16} and
a 15-layer U-net~\cite{RFB15a}.
The structures of these networks are shown in Fig.~\ref{fig:denoise_net}.

\begin{figure}
	\centering
	\includegraphics[width=1\linewidth]{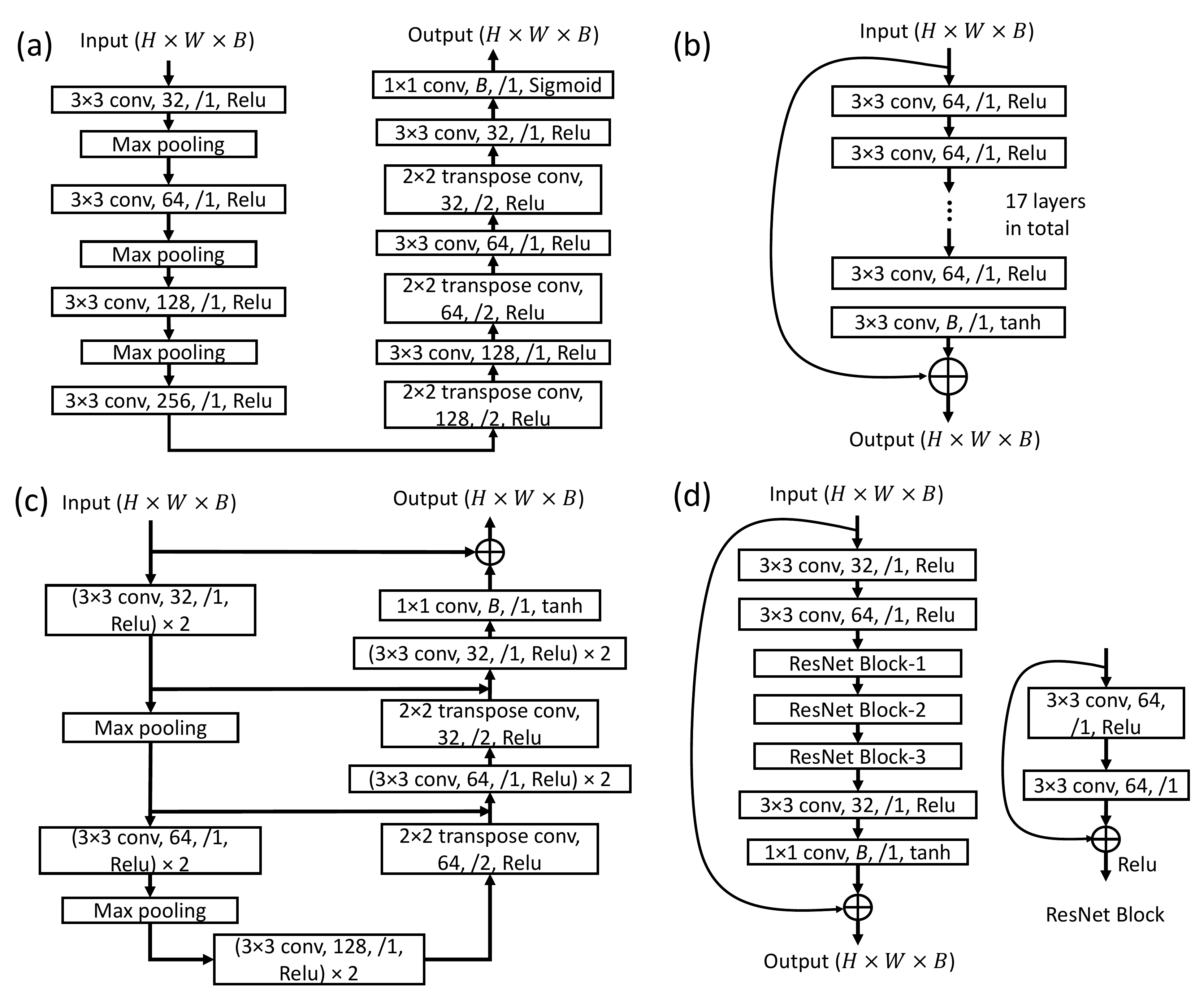}
	\caption{Denoiser network structures used in GAP-net: auto-encoder (a), DnCNN (b), U-net (c) and ResNet (d).}
	\label{fig:denoise_net}
\end{figure}

\subsection{Training Details}
\subsubsection{Training and testing datasets} For video SCI, our training dataset consists of 26000 videos generated from DAVIS2017~\cite{pont20172017} which has 90 scenes with 480p and 1080p resolution. We performed data augmentation by cropping and rotation. 

The synthetic testing data include six benchmark data with a size of $256\times256\times8$, including \texttt{Kobe}, \texttt{Runner}, \texttt{Drop}, \texttt{Traffic}, \texttt{Aerial} and \texttt{Vehicle} used in~\cite{Yuan2020_CVPR_PnP}.
The real data include three scenes (\texttt{Chopper Wheel} of size $256 \times 256 \times 14$~\cite{Patrick13OE}, \texttt{Water Balloon} and \texttt{Dominoes} of size  $512 \times 512 \times 10$~\cite{Qiao2020_APLP}.

For Spectral SCI, our training dataset  includes 4000 spectral cubes generated from CAVE~\cite{yasuma2010generalized} which have 30 scenes with a size of $512\times512\times31$ ranging over 400 to 700nm. We use  spectral interpolation to unify the datasets to the same wavelengths that our used in our experiments (28 channels from 450 to 650nm)~\cite{Meng20ECCV_TSAnet}. For model training, we did data augmentation by cropping, scaling, rotation and concatenation among the samples.

The synthetic testing data include 10 scenes cropped from KAIST~\cite{choi2017high} with the size of $256 \times 256 \times 28$, as shown in Fig.~\ref{fig:cassi_scene}. The real data include four scenes with a size of  $550 \times 550 \times 28$, including \texttt{Strawberry} and \texttt{Lego plant}, as well as \texttt{Lego man} and \texttt{Real plant}. 

\begin{figure}
	\centering
	\includegraphics[width=.8\linewidth]{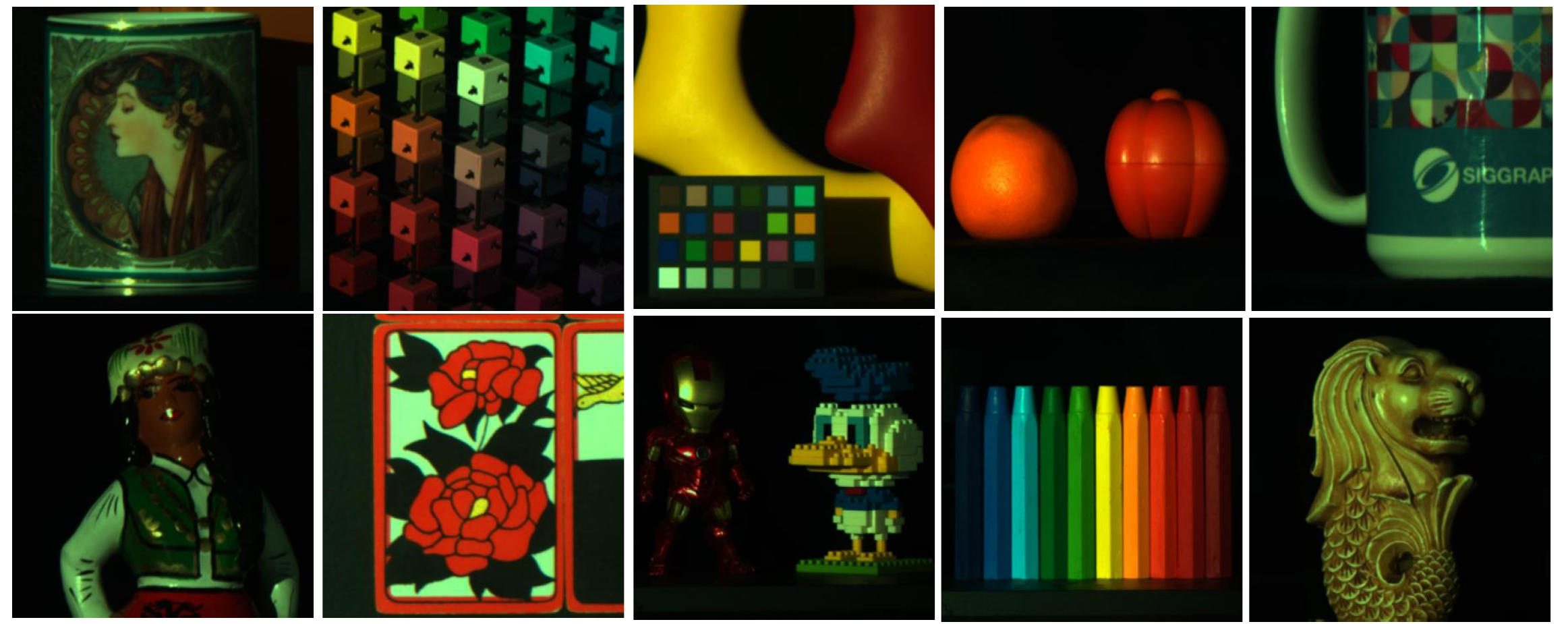}
	\caption{RGB images of 10 testing scenes for spectral SCI simulation.}
	\label{fig:cassi_scene}
\end{figure}

\subsubsection{Loss function}
For both video and spectral SCI, 
the loss function is the root mean square error (RMSE) between the ground truth and the output of the last stage
\begin{equation}
{
	{\cal L} = {\|\xv^*-\vv^{(K)}\|_2},
}
\end{equation}
where $\xv^*$ is the ground truth and $\vv^{(K)}$ denotes the output of the  last stage in GAP-net.

We also notice that for video SCI, during training, if we use the outputs of the last few stages, instead of only the last one,  the performance improves. For example, we can use the RMSE between the ground truth and the outputs of the last three stages of GAP-net with $K$ stages, which leads to
\begin{equation}
{\cal L} = \beta_{1}{\|\xv^*-\vv^{(K)}\|_2}+\beta_{2}{\|\xv^*-\vv^{(K-1)}\|_2}+\beta_{3}{\|\xv^*-\vv^{(K-2)}\|_2}.
\end{equation}
The parameters $\beta_{1},\beta_{2}$ and $\beta_{3}$ are set to  1, 0.5 and 0.5, respectively. During testing, we  only use the result from the last stage to measure the loss.
However, in our experiments in  spectral SCI, we did not find this approach effective.

\subsubsection{Implementation details}
We implemented GAP-net on both Pytorch and Tensorflow using NVIDIA 1080Ti GPUs, and trained all models by the Adam optimizer~\cite{kingma2014adam} in an end-to-end manner. The leaning rate is set to 0.001 initially, and scaled to $90\%$ of the previous one every 10 epochs for video SCI (every 30 epochs for spectral SCI).

\section{Experimental Results}
In this section, we show results of GAP-net compared with other algorithms for both video SCI and spectral SCI, on both simulation and real datasets.


\subsection{Video SCI}
We compare the performance of GAP-net in video SCI with that of   iterative algorithms (GAP-TV~\cite{Yuan16ICIP_GAP} and DeSCI~\cite{Liu19_PAMI_DeSCI}), end-to-end solutions (E2E-CNN~\cite{Qiao2020_APLP} and BIRNAT~\cite{Cheng20ECCV_Birnat}), and PnP method (PnP-FFDnet~\cite{Yuan2020_CVPR_PnP}). 
We also compare with Tensor-ADMM-net~\cite{Ma19ICCV} and Tensor-FISTA-net~\cite{Han_AAAI2020_FISTA}, though only partial (3 out of 6 scenes) results were shown in their papers.
Note that DeSCI is the state-of-art algorithm using optimization and BIRNAT is the most recent deep learning based approach leading to competitive result to DeSCI. 


\subsubsection{Synthetic data} 

\begin{table}[htbp!]
	\caption{\small Average PSNR in dB, SSIM and run time per measurement by different algorithms on six benchmark data of video SCI. NA: not available}
	\label{video_table}
	\centering
	\resizebox{1\textwidth}{!}
	{
		\begin{tabular}{c | c c c c c c | c c}
			\toprule
			\cmidrule(r){1-2}
			Methods     & Kobe     & Traffic & Runner & Drop & Vehicle & Aerial & Average & time (s)  \\
			\midrule
			GAP-TV~\cite{Yuan16ICIP_GAP} & 26.46, 0.885 & 20.89, 0.715 & 28.52, 0.909 & 34.63, 0.970 & 24.82 0.838 & 25.05, 0.828 & 26.73, 0.858 & 4.2     \\
			DeSCI~\cite{Liu19_PAMI_DeSCI}  & \textbf{33.25}, \textbf{0.952} & \underline{28.72}, 0.925 & \textbf{38.76}, 0.969 & \textbf{43.22}, \textbf{0.993} & 25.33, 0.860 & 27.04, 0.909  & 32.72, 0.935 & 6180     \\
			E2E-CNN~\cite{Qiao2020_APLP}  & 29.02, 0.861 & 23.45, 0.838 & 34.43, 0.958 & 36.77, 0.974 & 26.40, 0.886 & 27.52, 0.882 & 29.26, 0.900 & 0.023 \\
			BIRNAT~\cite{Cheng20ECCV_Birnat} & \underline{32.71}, \underline{0.950} & \textbf{29.33}, \textbf{0.942} & \underline{38.70}, \textbf{0.976} & \underline{42.28}, \underline{0.992} & \textbf{27.84}, \underline{0.927} & \textbf{28.99}, \textbf{0.917} & \textbf{33.31}, \textbf{0.951} & 0.16 \\
			PnP-FFDnet~\cite{Yuan2020_CVPR_PnP} & 30.50, 0.926 & 24.18, 0.828 & 32.15, 0.933 & 40.70, 0.989 & 25.42, 0.849 & 25.27, 0.829 & 29.70, 0.892 & 3.0\\
			Tensor-ADMM-net~\cite{Ma19ICCV} & 30.50, 0.890 & NA &  NA &   NA& 25.42, 0.780 & 25.27, 0.860 &  NA & 2.1\\
			Tensor-FISTA-net~\cite{Han_AAAI2020_FISTA} & 31.41, 0.920 & NA &  NA &  NA & 26.46, 0.890 & 27.46, 0.880 &  NA & 1.7\\
			\midrule
			GAP-net-AE-S9 & 24.20, 0.570 & 21.13, 0.685 & 29.18, 0.886 & 32.21, 0.907 & 24.19, 0.769 & 24.41, 0.744 & 25.89, 0.760 & 0.0036\\
			GAP-net-DnCNN-S9 & 31.09, 0.930 & 27.36, 0.912 & 37.49, 0.972 & 41.52, 0.990 & 27.57, 0.925 & 28.56, 0.906 & 32.27, 0.939 & 0.0081\\
			GAP-net-ResNet-S9 & 31.57, 0.938 & 27.61, 0.918 & 37.70, 0.972 & 41.75, 0.991 & 27.68, 0.925 & 28.74, 0.910 & 32.51, 0.942 & 0.0040\\
			ADMM-net-Unet-S9 & 31.87, 0.941 & 27.88, 0.923 & 37.75, 0.973 & 41.41, 0.991 & 27.58, 0.923 & 28.70, 0.910 & 32.53, 0.943 & 0.0058\\
			GAP-net-Unet-S9 & 31.76, 0.941 & 27.87, 0.924 & 37.89, 0.974 & 41.43, 0.991 & 27.53, 0.926 & 28.57, 0.909 & 32.51, 0.944 & 0.0052\\
			\hline
			GAP-net-Unet-S12 & 32.09, 0.944 & 28.19, \underline{0.929} & 38.12, \underline{0.975} & 42.02, \underline{0.992} & \underline{27.83}, \textbf{0.931} & \underline{28.88}, \underline{0.914} & \underline{32.86}, \underline{0.947} & 0.0072\\
			\bottomrule
		\end{tabular}
	}
\end{table}

\begin{figure}
	\includegraphics[width=1\linewidth]{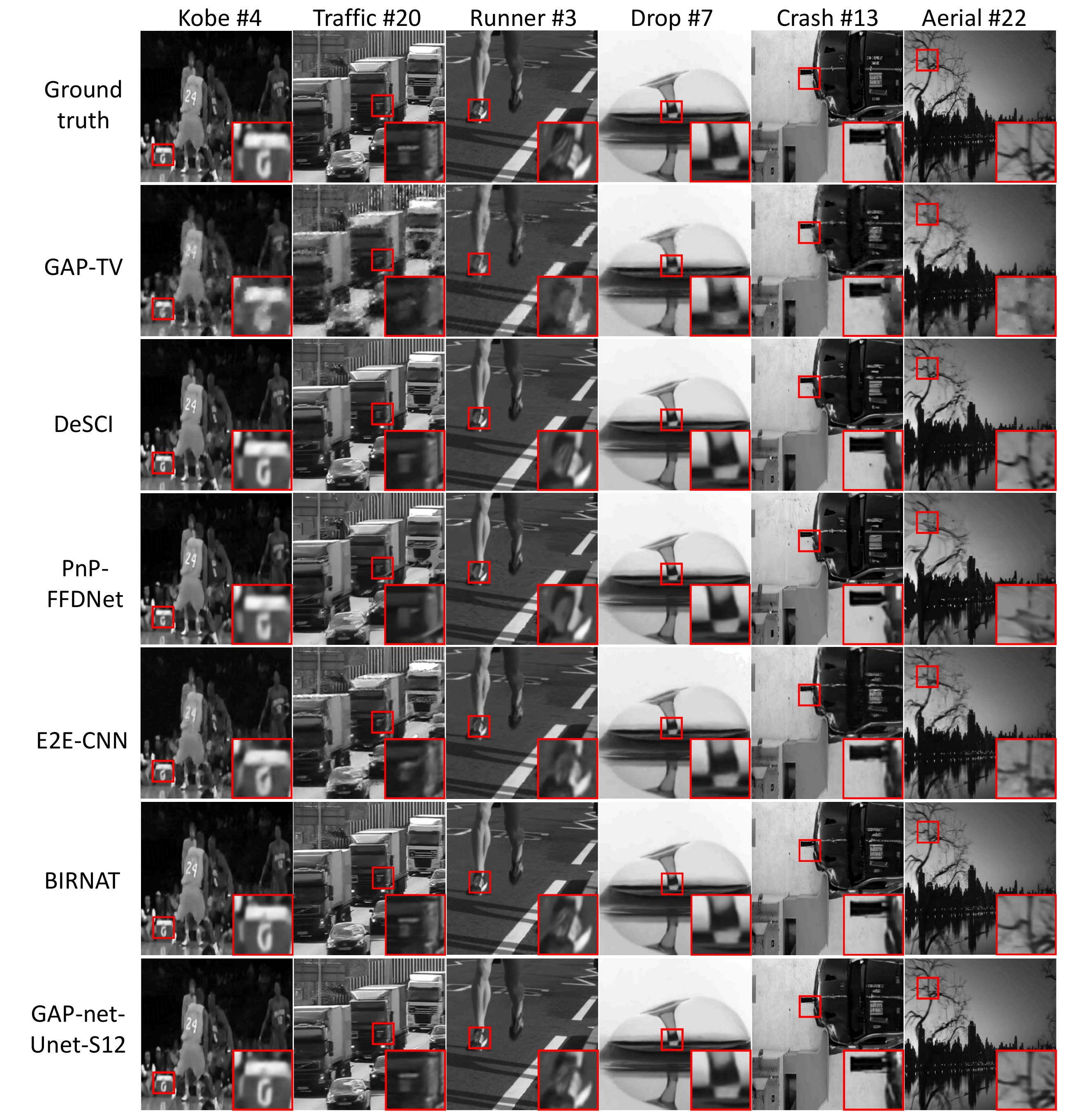}
	\caption{\small Video SCI: reconstructed frames by six algorithms in six synthetic benchmark data.}
	\label{fig:video_result}
\end{figure}

Table~\ref{video_table} summarizes the PSNR and SSIM~\cite{wang2004image} results of different algorithms. For GAP-net, we use $K=9$ stages and    report the results   for different denoisers (14-layer AEs, 17-layer DnCNN, 15-layer U-net and 10-layer ResNet) and compare them with the results of an ADMM-net with 9 stages.
It can be seen that, except for the AE-based GAP-net, the proposed GAP-net with 9 stages, especially   using U-net or ResNet, outperforms GAP-TV, E2E-CNN and PnP methods.

Since the skip-connection in U-net can be recognized as residual learning, this indicates that residual learning is necessary for the unfolding denoiser (compared with the original AEs). For the AE-based case, we believe that its relative poor performance is due to the known challenge in training high-performance AEs, \ie, AEs with low distortion. This leads for instance to having large values of $\delta_k$ in Theorem \ref{The:GAP_SCI_pro_2} and a resulting large final distortion. Tensor-ADMM-net and Tensor-FISTA-net only reported results for 3 datasets in \cite{Ma19ICCV} and \cite{Han_AAAI2020_FISTA}, and all of them are worse than GAP-net's results.%

Table~\ref{video_table} also shows that GAP-net is much faster than other algorithms. Specifically, it only needs 4ms (using ResNet) to provide good results, which along with the SCI cameras, can perform real-time capture and reconstruction of up to 250 measurements per second.
We further trained a deeper GAP-net with 12 stages U-net, and its result outperforms the state-of-the-art optimization algorithm DeSCI on average. 
In addition, ADMM-net and GAP-net with the same network structure provide similar results, but GAP-net spends ($0.6$ms) less time and thus is more efficient. {Note that though the most recent BIRNAT provides 0.5dB higher PSNR than ours, the running time is more than 20 times longer than our GAP-net. This will limit the running speed of the end-to-end SCI systems in real applications.}

\begin{figure}
	\centering
	\includegraphics[width=.6\linewidth]{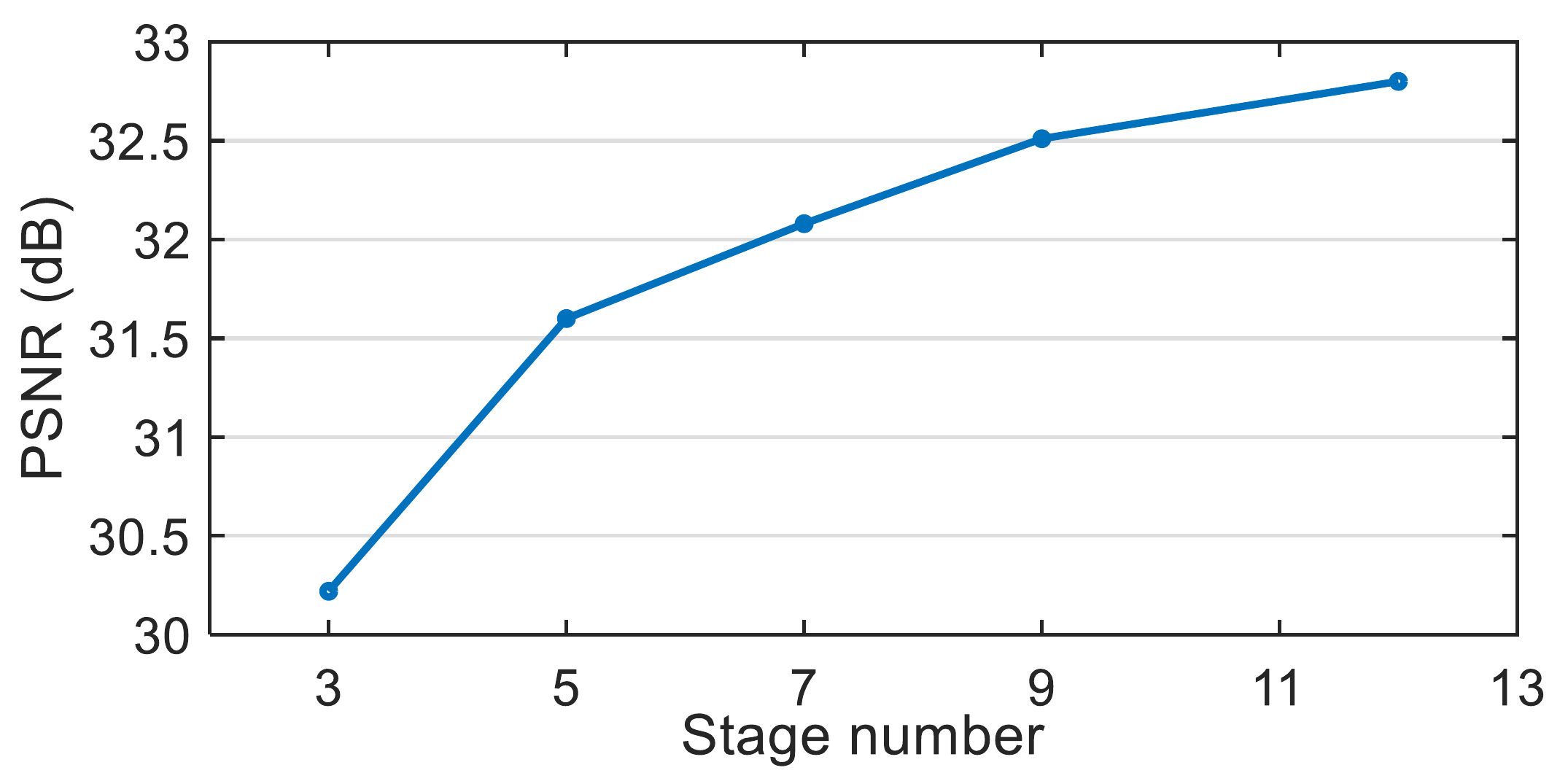}
	\caption{\small Impact of the stage number on the reconstruction accuracy of video SCI.}
	\label{fig:stage_number}
\end{figure}

Fig.~\ref{fig:video_result} plots selected reconstructed frames of the six datasets using different algorithms. We  observe that though DeSCI shows smoother results, our proposed GAP-net provide more accurate details in some cases. The difference between our GAP-net and BIRNAT is marginal and difficult to notice.

To verify our theoretical derivation, Fig.~\ref{fig:stage_number} plots how average PSNR improves as $K$ increases in the U-net-based GAP-net, which is consistent with our theoretical results. We also notice that a 3-stage GAP-net can provide better results than most optimization-based algorithms except 
for DeSCI. 


\subsubsection{Real data} 
We apply the proposed GAP-net (9 stages with a 15-layer U-net) to real data captured by SCI cameras~\cite{Patrick13OE,Qiao2020_APLP}. We compare our results with other five algorithms (GAP-TV, DeSCI, PnP-FFDnet, E2E-CNN and BIRNAT). 
Fig.~\ref{fig:chop} and Fig.~\ref{fig:waterball} demonstrate the reconstructed results of \texttt{Chopper Wheel} of size $256 \times 256 \times 14$ as well as \texttt{Water Balloon} and \texttt{Dominoes} of size  $512 \times 512 \times 10$, respectively.
It can be observed that DeSCI and PnP-FFDnet show smoother results but with a few artifacts. 
The results of E2E-CNN are not clean on the motion part of scenes.

The results of GAP-net show sharper edges of motion objects and extremely clean static details and background without artifacts, which is comparable or even better than the results of BIRNAT. 
Specifically, only GAP-net can recover complete and sharp edges of the letter "D" on \texttt{Chopper Wheel} (zoomed region in Fig.~\ref{fig:chop}). In addition, GAP-net can produce cleaner background and object details (zoomed regions in Fig.~\ref{fig:waterball}).

\begin{figure}[htbp!]
	\centering
	\includegraphics[width=1\linewidth]{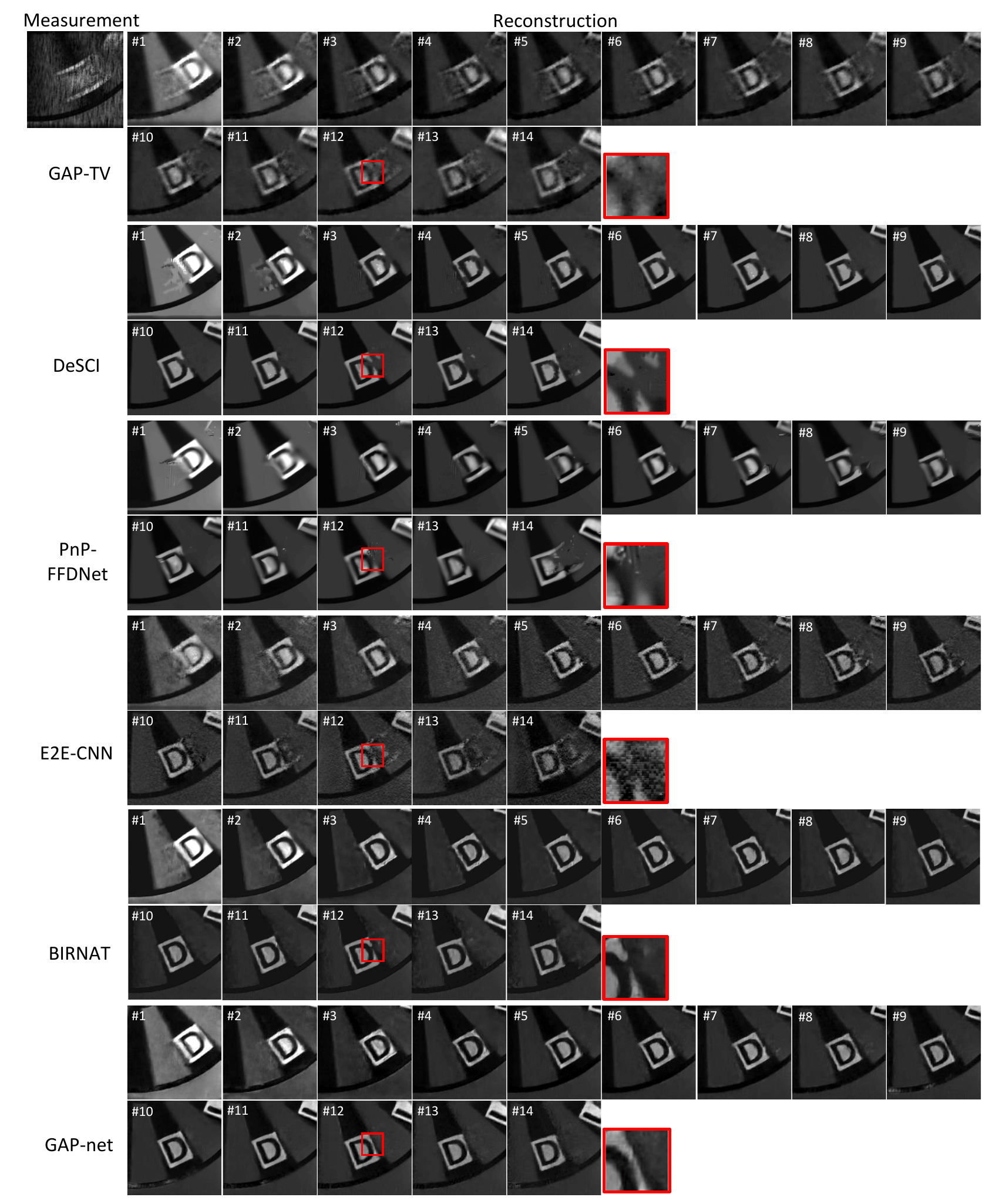}
	\caption{Video SCI: reconstructed results of different algorithms on \texttt{Chop Wheel} [real data].}
	\label{fig:chop}
\end{figure}


\begin{figure}[htbp!]
	\centering
	\includegraphics[width=1\linewidth]{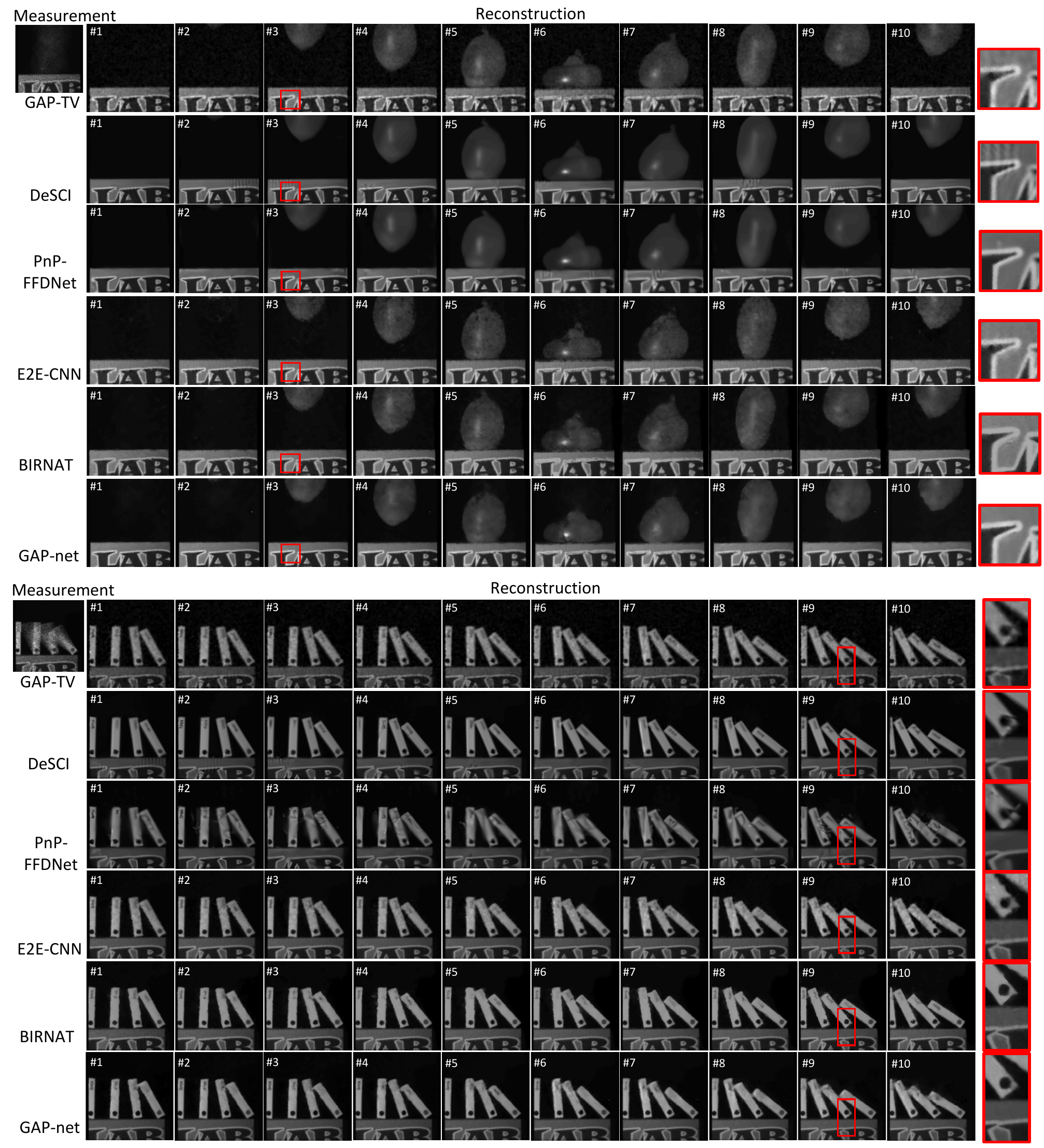}
	\caption{Video SCI: reconstructed results of different algorithms on \texttt{Water Balloon} (top) and \texttt{Dominoes} (bottom) [real data].}
	\label{fig:waterball}
\end{figure}


\subsection{Spectral SCI}
For spectral SCI, following~\cite{Wagadarikar08CASSI} we use the CASSI hardware system we have built in~\cite{Meng20ECCV_TSAnet},  which can produce spectral images with 28 spectral channels ranging from 450 to 650nm with hardware details shown in~\cite{Meng20ECCV_TSAnet}. We compare the proposed GAP-net with iteration algorithms: TwIST~\cite{Bioucas-Dias2007TwIST}, GAP-TV~\cite{Yuan16ICIP_GAP} \& DeSCI~\cite{Liu19_PAMI_DeSCI}, E2E-CNN: $\lambda$-net~\cite{Miao19ICCV}, TSA-net~\cite{Meng20ECCV_TSAnet} and an unfolding method: HSSP~\cite{Wang19_CVPR_HSSP} on both synthetic and real data.

\begin{table}
	\caption{\small Average PSNR in dB, SSIM and run time per measurement by seven algorithms on 10 scenes for spectral SCI.}
	\label{spectral_table}
	\centering
	\resizebox{1\textwidth}{!}
	{
		\begin{tabular}{c | c c c c c c c }
			\toprule
			\cmidrule(r){1-2}
			Methods     & TwIST     & GAP-TV & DeSCI & HSSP & $\lambda$-net & TSA-net & GAP-net  \\
			\midrule
			Scene1 & 24.81, 0.730 & 25.13, 0.724 & 27.15, 0.794 & 31.07, 0.852 & 30.82, 0.880 & 31.26, 0.887 & {\bf 33.03}, {\bf 0.921} \\
			Scene2 & 19.99, 0.632 & 20.67, 0.630 & 22.26, 0.694 & 26.30, 0.798 & 26.30, 0.846 & 26.88, 0.855 & {\bf 29.52}, {\bf 0.903} \\
			Scene3 & 21.14, 0.764 & 23.19, 0.757 & 26.56, 0.877 & 29.00, 0.875 & 29.42, 0.916 & 30.03, 0.921 & {\bf 33.04}, {\bf 0.940}  \\
			Scene4 & 30.30, 0.874 & 35.13, 0.870 & 39.00, 0.965 & 38.24, 0.926 & 36.27, 0.962 & 39.90, 0.964 & {\bf 41.59}, {\bf 0.972}  \\
			Scene5 & 21.68, 0.688 & 22.31, 0.674 & 24.80, 0.778 & 27.98, 0.827 & 27.84, 0.866 & 28.89, 0.878 & {\bf 30.95, 0.924}  \\
			Scene6 & 22.16, 0.660 & 22.90, 0.635 & 23.55, 0.753 & 29.16, 0.823 & 30.69, 0.886 & 31.30, 0.895 & {\bf 32.88, 0.927} \\
			Scene7 & 17.71, 0.694 & 17.98, 0.670 & 20.03, 0.772 & 24.11, 0.851 & 24.20, 0.875 & 25.16, 0.887 & {\bf 27.60, 0.921} \\
			Scene8 & 22.39, 0.682 & 23.00, 0.624 & 20.29, 0.740 & 27.94, 0.831 & 28.86, 0.880 & 29.69, 0.887 & {\bf 30.17, 0.904}  \\
			Scene9 & 21.43, 0.729 & 23.36, 0.717 & 23.98, 0.818 & 29.14, 0.822 & 29.32, 0.902 & 30.03, 0.903 & {\bf 32.74, 0.927} \\
			Scene10 & 22.87, 0.595 & 23.70, 0.551 & 25.94, 0.666 & 26.44, 0.740 & 27.66, 0.843 & 28.32, 0.848 & {\bf 29.73, 0.901} \\
			\midrule
			Average & 22.44, 0.703 & 23.73, 0.683 & 25.86, 0.785  & 28.93, 0.834 & 29.25, 0.886 & 30.15, 0.893 & {\bf 32.13, 0.924}  \\
			Time (s) & 22.2 & 14.5 & 8500 & 0.011 & 0.013 & 0.03 & 0.016 \\
			\bottomrule
		\end{tabular}
	}
\end{table}

\begin{figure}
	\includegraphics[width=1\linewidth]{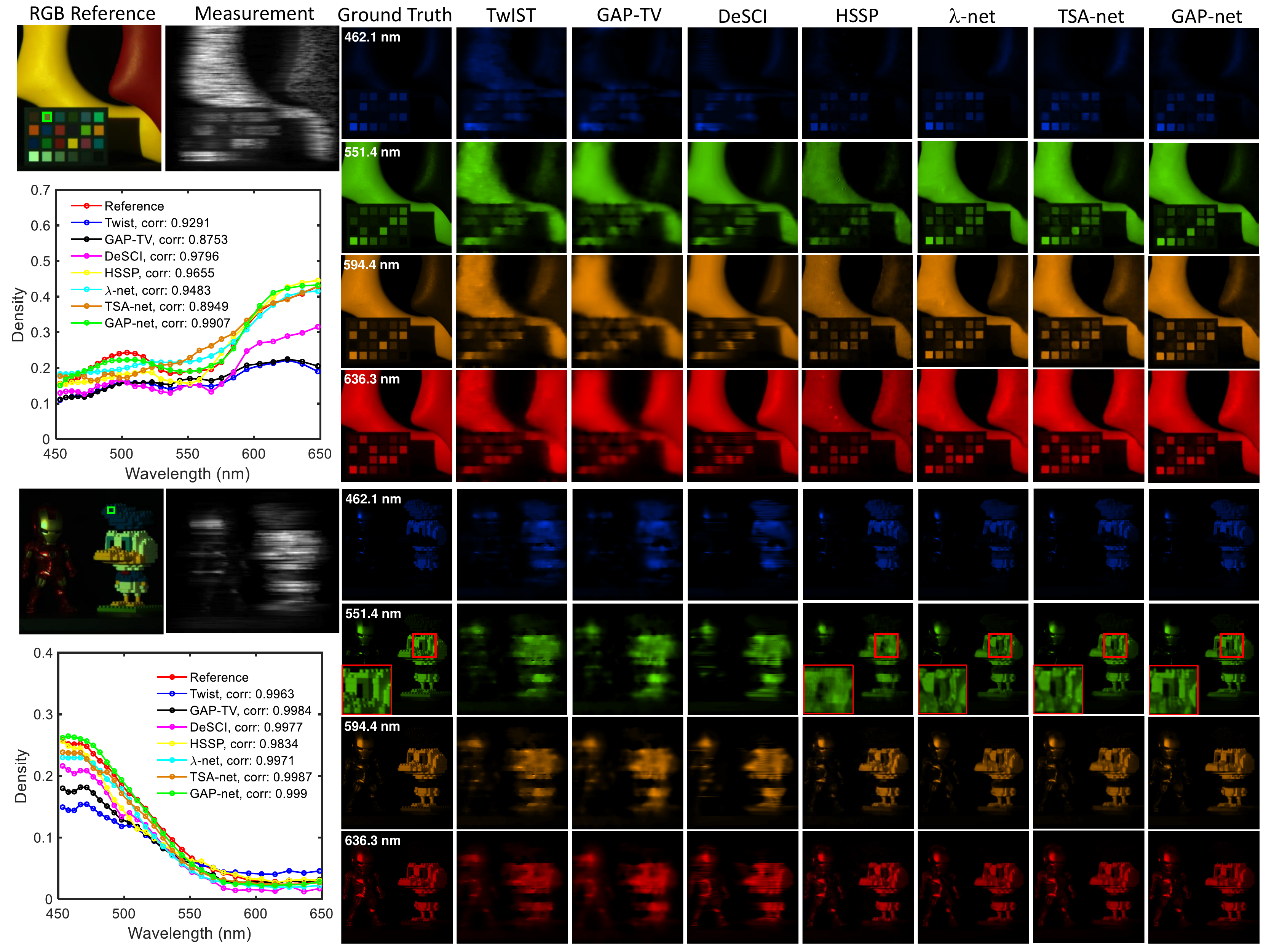}
	\caption{\small Spectral SCI: reconstructed images and spectra of scene 3 and scene 8 in Fig.~\ref{fig:cassi_scene} with seven algorithms in synthetic data.
	}
	\label{fig:spectral_simu1}
\end{figure}
\subsubsection{Synthetic data} 
We cropped a $256 \times 256$ region of the real captured mask as the mask for simulation. The measurements are generated by this mask and hyperspectral datasets mentioned before, and the shift of two adjacent channels is two pixels. We trained the GAP-net with 9 stages using U-net as the denoiser for reconstruction. Though we believe more stages of GAP-net will provide better results, due to large-scale of the data (28 spectral channels) and the GPU memory, we only train a 9-stage network to demonstrate the performance of our proposed GAP-net here. Importantly, this already leads to state-of-the-art results on spectral SCI.

Table~\ref{spectral_table} lists the PSNR and SSIM of the reconstructed results of 10 testing scenes shown in Fig.~\ref{fig:cassi_scene}  by seven algorithms  including the most recent TSA-net~\cite{Meng20ECCV_TSAnet}.
GAP-net achieves a significant improvement in reconstruction, \ie, the average PSNR is 2dB higher than TSA-net which is the best among other algorithms. 

The reconstructed results of two scenes, including images with four channels and recovered spectra of the selected regions, are shown in Fig.~\ref{fig:spectral_simu1}. It can be seen that the results of iterative optimization algorithms suffer from blurry artifacts resulted from the coded measurements, which is due to the disperser in the hardware system. Among deep learning-based methods, the proposed GAP-net can reconstruct sharper spatial details and more accurate spectra than other methods. 

Fig.~\ref{fig:spectral_simu2} compares the reconstructed results of seven algorithms on 4 scenes, and the PSNR and SSIM values are provided for each result. To visualize the recovered color, we convert the spectral images to synthetic-RGB (sRGB) via the CIE (International Commission on Illumination) color matching function~\cite{smith1931cie}. It can be observed that GAP-net outperforms other algorithms in both spatial details and spectral accuracy. Clear details and sharp edges can be recovered.
Please refer to the zoomed regions of each scene.

\begin{figure}
	\includegraphics[width=1\linewidth]{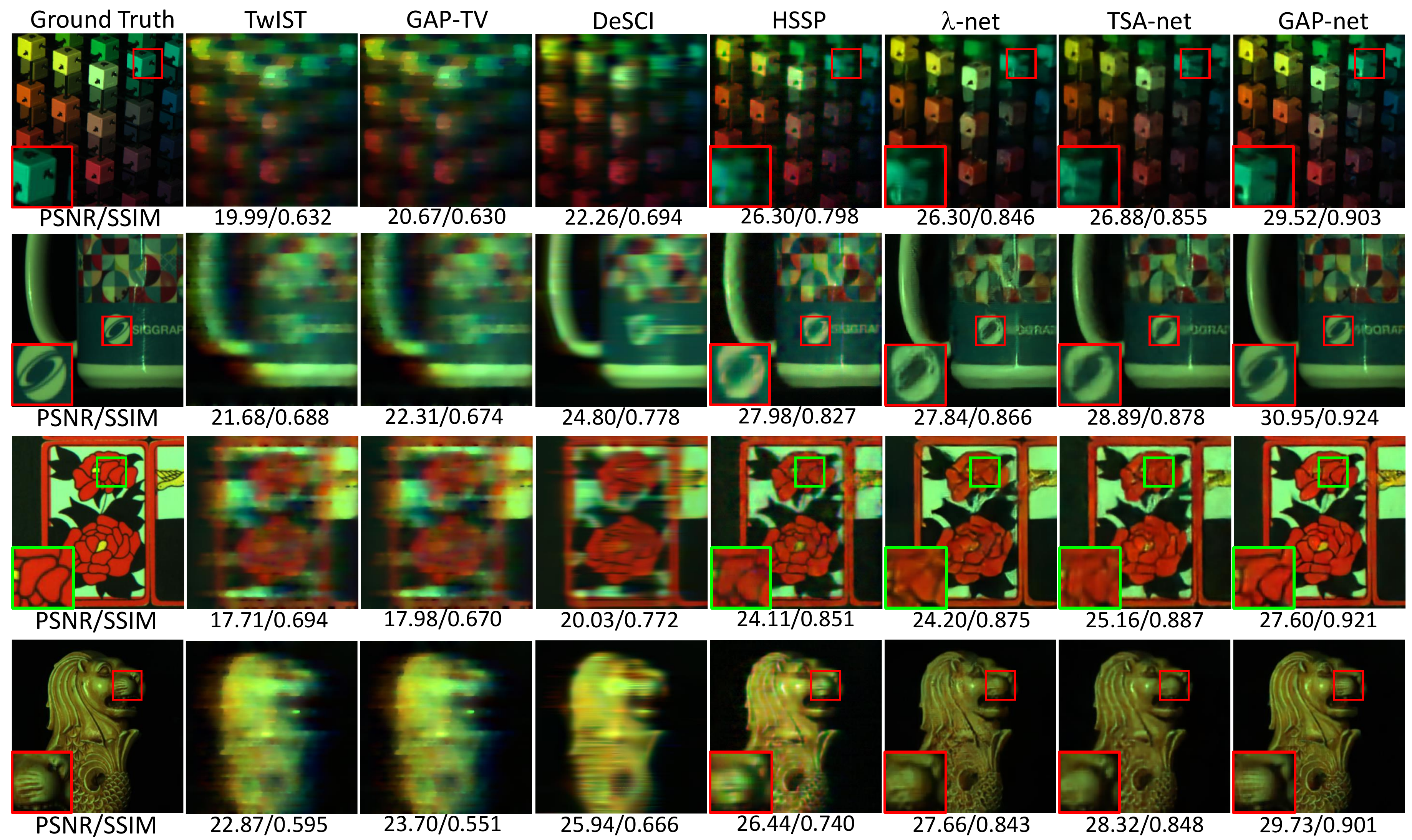}
	\caption{\small Spectral SCI: reconstructed images (sRGB) with 7 algorithms in synthetic data.
	}
	\label{fig:spectral_simu2}
\end{figure}

\subsubsection{Flexibility of GAP-net to Mask Modulation}

The proposed GAP-net treats the CNN in each stage as a denoising network, and thus  GAP-net is flexible with respect to the signal modulation, \ie, the masks used in the SCI system.
To verify this point, we conduct simulations of spectral SCI by training GAP-net using one mask and testing it on the other three masks.
Table~\ref{spectral_table2} lists the average testing results of the ten scenes using three new masks that are cropped from the real captured mask with different regions from the mask used in training. We  observe that the image quality degrades less than 1 dB in PSNR when using new masks for testing, and the results are still better than other existing algorithms. Therefore, due to the flexibility of masks, a well trained GAP-net on small-scale data can be used for large-scale reconstruction using patch-based testing, and the trained network can also be applied to different systems.

\begin{table}
	\caption{\small Testing results (PSNR in dB and SSIM) of spectral SCI using different masks}
	\label{spectral_table2}
	\centering
	\resizebox{0.45\textwidth}{!}
	{
		\begin{tabular}{c | c }
			\toprule
			\cmidrule(r){1-2}
			Mask for testing & PSNR/SSIM    \\
			\midrule
			Mask used in training & 32.13, 0.924 \\
			New mask 1 & 31.36, 0.901 \\
			New mask 2 & 31.05, 0.895 \\
			New mask 3 & 31.60, 0.902 \\
			\bottomrule
		\end{tabular}
	}
\end{table}

\begin{figure}[htbp!]
	\centering
	\includegraphics[width=1\linewidth]{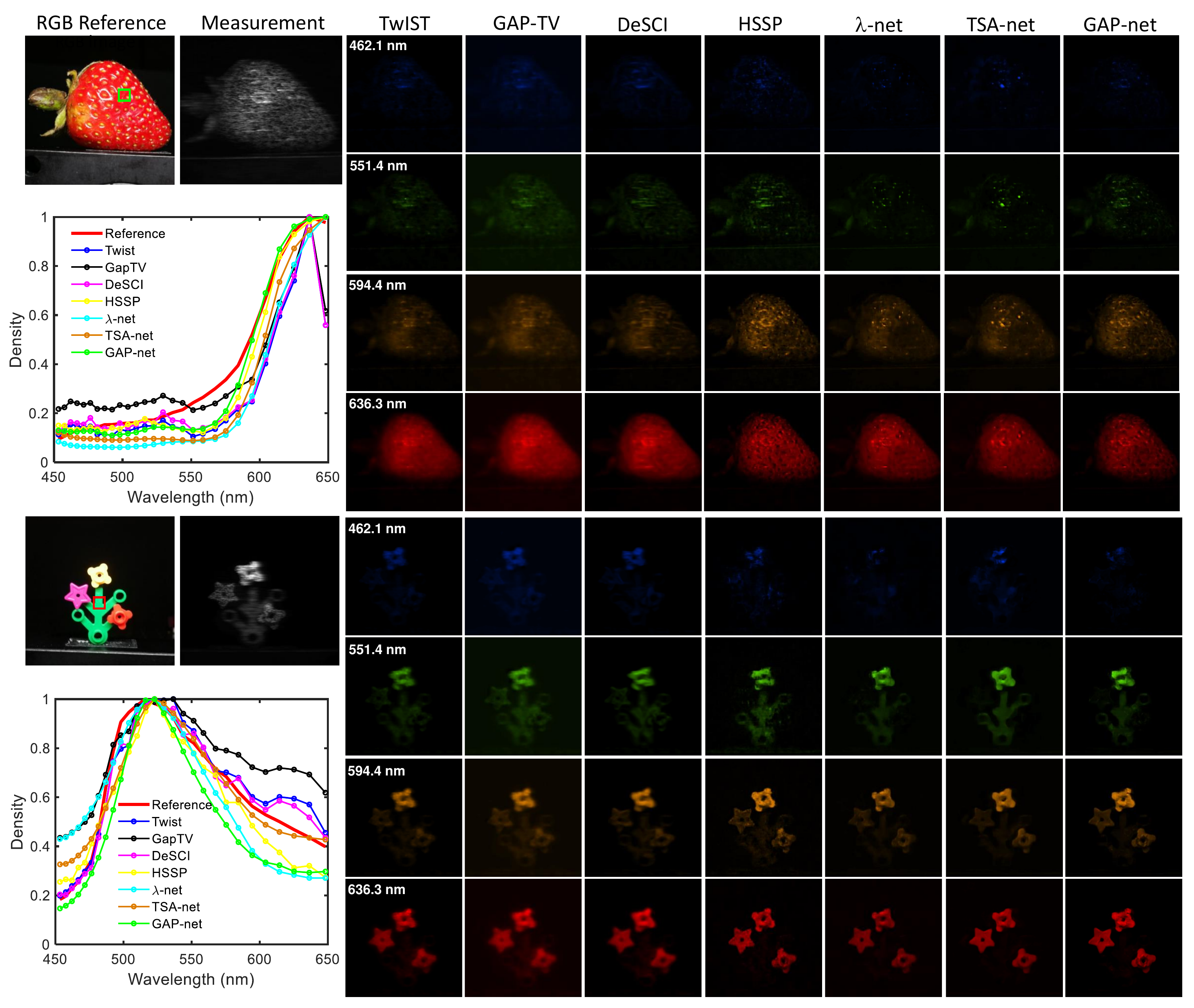}
	\caption{ Spectral SCI: reconstructed images with 7 algorithms on \texttt{Strawberry} (top) and \texttt{Lego plant} (bottom) [real data].}
	\label{fig:spectral_real1}
\end{figure}


\begin{figure}[htbp!]
	\centering
	\includegraphics[width=1\linewidth]{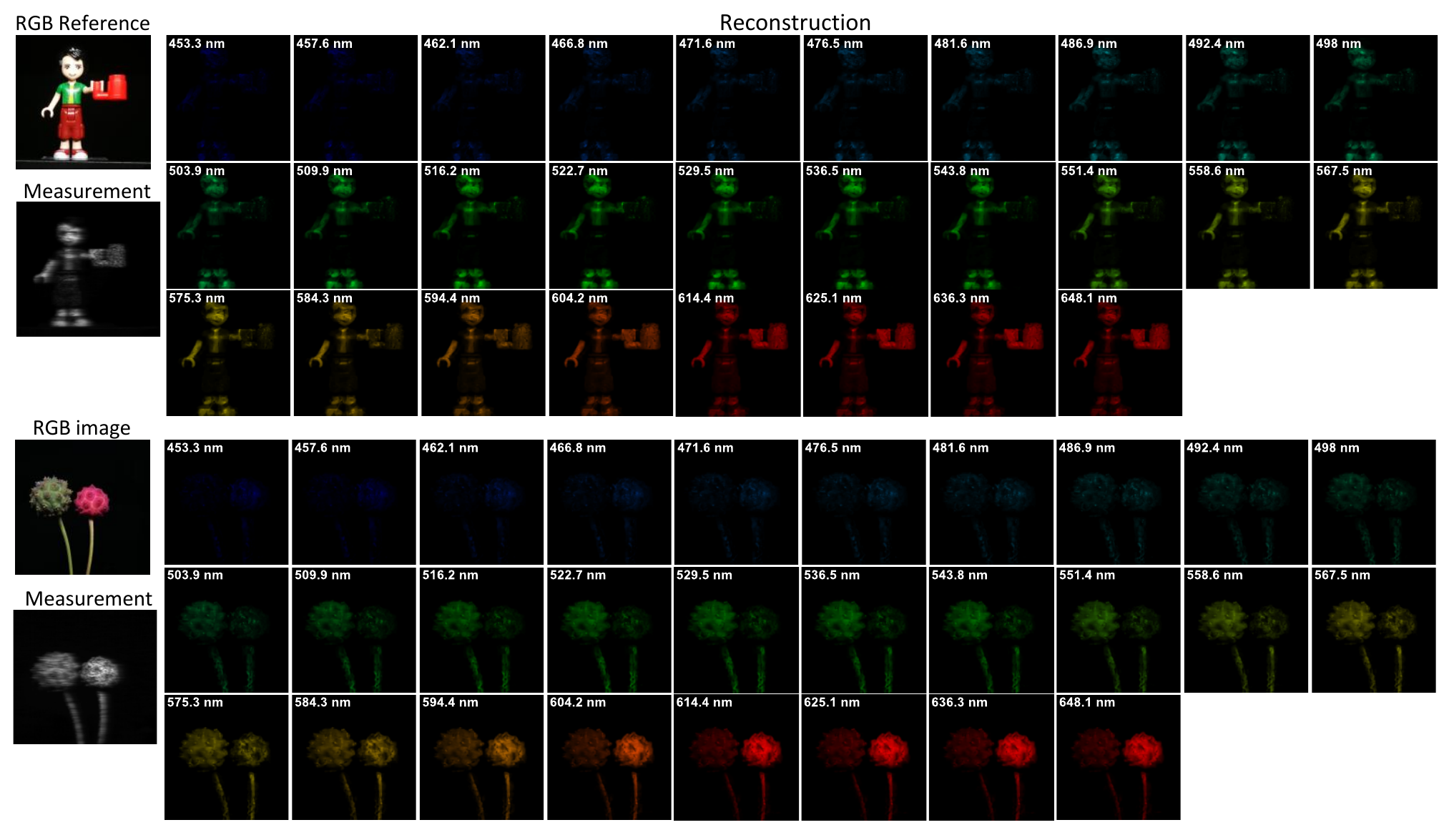}
	\caption{ Spectral SCI: reconstructed images with 28 channels by GAP-Net on \texttt{lego man} (top) and \texttt{Real plant} (bottom) [real data].}
	\label{fig:spectral_real2}
\end{figure}

\subsubsection{Real data} 
For real data testing, we use the {pre-trained GAP-net on simulation} to reconstruct the scenes with larger scale by patch-based manner. We reconstructed four scenes with a size of $550 \times 550 \times 28$. 

Fig.~\ref{fig:spectral_real1} shows the results of \texttt{Strawberry} and \texttt{Lego plant} with four spectral channels and recovered spectra of the selected regions. We also compare our method with other six algorithms. We can observe that the results of deep learning-based methods are much better than the results of iterative algorithms. Our GAP-net can provide more details of the objects and higher spectral contrast than other deep learning based methods. Fig.~\ref{fig:spectral_real2} shows the reconstructed images of \texttt{Lego man} and \texttt{Real plant} with 28 spectral channels. The results provide both small-scale fine details and large-scale sharp edges.








%
%

\section{Proof of the main result}\label{Appex:1}
Recall that in Theorem~\ref{The:GAP_SCI_pro_2}, we assume that the non-zero entries of the sensing matrix $\Hmat$ are drawn i.i.d.~$\Nc(0,1)$. In this setting, to account for the scaling of matrix $\Hmat$, we update the first step of GAP-net, \ie, $\xv^{(k+1)} =  \vv^{(k)} + \Hmat\ts (\Hmat \Hmat\ts)\inv (\yv - \Hmat \vv^{(k)})$ to
\[
\xv^{(k+1)} =  \vv^{(k)} + B\Hmat\ts (\Hmat \Hmat\ts)\inv (\yv - \Hmat \vv^{(k)}).
\]
The only difference here is the factor $B$, which takes care of the scaling of the variables. 

In the following,  we first derive a more general result on the convergence performance of GAP-net. Then, Theorem~\ref{The:GAP_SCI_pro_2}  follows as a corollary of this result. 
\begin{theorem} \label{The:GAP_SCI_pro}
	Consider the sensing model of SCI. Assume that generative function $g_k: {\cal U}^{\eta_k} \rightarrow {\cal X}^{nB}$ covers set ${\cal Q}$ with distortion $\delta_k$. Further assume that $\delta_K\le \dots \le \delta_1$ and $\eta_1\le \dots \le \eta_K$. 
	Let $b_k\in\mathds{N}^+$, $k=0,\ldots,K$, denote auxiliary variables. Define 
	\begin{equation}
	\gamma_k\stackrel{\rm def}{=}{L_k 2^{-b_k}\sqrt{\eta_k}\over \delta_k\sqrt{nB}}
	\label{Eq:gammak}
	\end{equation}	
	and \begin{equation}
	\alpha_k\stackrel{\rm def}{=} 2(1+2{B})(\gamma_k(1+{1\over 1-\gamma_k}) )^{1\over 2}. \label{Eq:alphak}
	\end{equation}
	Then, if ${1 \over \sqrt{nB}}\|\vv^{(k)}-\vvt_k\|\geq \delta_k$ and ${1 \over \sqrt{nB}}\|\vv^{(k-1)}-\vvt_k\|\geq \delta_k$, we have
	\begin{align}
	{1\over \sqrt{nB}}\|\vvt_k- \vv^{(k)}\|_2
	\leq &{2\over \sqrt{nB}}(\lambda + \alpha_k) \|\vvt_{k-1}- \vv^{(k-1)}\|_2 \nonumber\\
	& \quad+(\lambda + \alpha_k)(\delta_k+\delta_{k-1})+2{B}\delta_k.
	\end{align}	
	with a probability larger than
	\[
	1-\sum_{k=0}^{K-1} |\Uc^{(k)}_Q|^{2\eta_k}\exp\Big(-{2\lambda^2 n\delta_k^4(1-\gamma_k)^4\over  4 {B}^2\rho^4 }\Big),
	\]
	where $\Uc^{(k)}_Q=\{[u]_{b_k}: u\in\Uc\}$ and $\lambda\in (0, 0.5-\alpha_k)$.
\end{theorem}
\begin{proof}[Proof of Theorem \ref{The:GAP_SCI_pro}] 
	By using the definition of distortion based generative model, the update equations of GAP-net are 
	\begin{eqnarray}
	\xv^{(k+1)} &=&  \vv^{(k)} + B\Hmat\ts (\Hmat \Hmat\ts)\inv (\yv - \Hmat \vv^{(k)}), \label{Eq:Xk_gau}\\
	{\vv}^{(k)} &=& g_k(\fv^{(k)}), ~~{\text{with}}~~   \fv^{(k)} = \argmin_{\fv \in {\mathbb R}^{\eta_k}}\|\xv^{(k)} -g_k(\fv)\|_2.
	\end{eqnarray}
Since by assumption, $g_k$ covers $\Qc$ with distortion $\delta_k$, we have 
	${1\over \sqrt{nB}}\|\vv^{(k)}-\vvt_k\| \leq \delta_k.$
	Moreover, since $\fv^{(k)}=\argmin_{\fv\in\mathds{R}^{\eta_k}}\|\xv^{(k)}- g_{k}(\fv)\|$, it follows that $\|\xv^{(k)}- \vv^{(k)}\|\leq \|\xv^{(k)}-\vvt_k\|$. 
	But, $
	    \|\xv^{(k)}- \vv^{(k)}\|^2=\|\xv^{(k)}-\vvt_k+\vvt_k- \vv^{(k)}\|^2 \nonumber\\
	    =\|\xv^{(k)}-\vvt_k\|^2+\|\vvt_k- \vv^{(k)}\|^2+2\langle \xv^{(k)}-\vvt_k,\vvt_k- \vv^{(k)}\rangle.$
	Therefore, $\|\vvt_k- \vv^{(k)}\|^2+2\langle \xv^{(k)}-\vvt_k,\vvt_k- \vv^{(k)}\rangle\leq 0$, or
	\begin{align}
	\|\vvt_k- \vv^{(k)}\|^2\leq 2\langle \vvt_k-\xv^{(k)},\vvt_k- \vv^{(k)}\rangle.
	\end{align}
	Since $\xv^{(k)} =  \vv^{(k-1)} + B\Hmat\ts \Rmat\inv \Hmat(\xv^* - \vv^{(k-1)}) $, we have
	\begin{align*}
	\langle\vvt_k-\xv^{(k)},\vvt_k- \vv^{(k)}\rangle &= \langle \vvt_k - \vv^{(k-1)} - B\Hmat\ts \Rmat\inv \Hmat(\xv^* - \vv^{(k-1)}), \vvt_k- \vv^{(k)}\rangle \nonumber\\
	& = \langle \vvt_k - \vv^{(k-1)}, \vvt_k- \vv^{(k)}\rangle \nonumber\\
	& \quad-B  \langle \Rmat\inv \Hmat(\xv^* - \vv^{(k-1)}),  \Hmat (\vvt_k- \vv^{(k)}) \rangle.
	\end{align*}
	Also 
	\begin{align*}
	&\langle \Rmat\inv \Hmat(\xv^* - \vv^{(k-1)}),  \Hmat (\vvt_k- \vv^{(k)}) \rangle \nonumber\\
	&= \langle \Rmat\inv \Hmat(\vvt_k  - \vv^{(k-1)}),  \Hmat (\vvt_k- \vv^{(k)}) \rangle + \langle \Rmat\inv \Hmat(\xv^* - \vvt_k),  \Hmat (\vvt_k- \vv^{(k)}) \rangle.
	\end{align*}
	Therefore, in summary,
	\begin{align}
	\|\vvt_k- \vv^{(k)}\|^2\;\leq &\; 2\langle \vvt_k - \vv^{(k-1)}, \vvt_k- \vv^{(k)}\rangle \nonumber \\
	&- 2B \langle \Rmat\inv \Hmat(\vvt_k  - \vv^{(k-1)}),  \Hmat (\vvt_k- \vv^{(k)}) \rangle \nonumber \\
	&-2B \langle \Rmat\inv \Hmat(\xv^* - \vvt_k),  \Hmat (\vvt_k- \vv^{(k)}) \rangle \label{Eq:vvtmvvk}
	\end{align}
	Define
	\begin{align}
	\ev = \frac{\vvt_k - \vv^{(k-1)}}{\|\vvt_k - \vv^{(k-1)}\|}, \quad \ev' = \frac{\vvt_k- \vv^{(k)}}{\|\vvt_k- \vv^{(k)}\|}.
	\end{align}
	Using these definitions and applying Cauchy-Schwartz inequality to the last term in \eqref{Eq:vvtmvvk}, we have
	\begin{align}
	\|\vvt_k- \vv^{(k)}\| \leq &2 |\langle\ev, \ev'\rangle -B \langle\Rmat\inv \Hmat \ev, \Hmat \ev'\rangle| \|\vvt_k - \vv^{(k-1)}\| \nonumber\\
	&+  2 B \sigma_{\max} (\Hmat\ts\Rmat\inv \Hmat) \|\xv^* - \vvt_k\|, \label{eq:error-reduction-main-1}
	\end{align}
	where $\sigma_{\max} (~)$ is the maximum eigenvalue of the ensued matrix.
	
	For $k=1,\ldots,K$, let $\uv_Q^{(k)}\triangleq [\uv^{(k)}]_b$ and $\vv_Q^{(k)}\triangleq g_k(\uv_Q^{(k)})$ and define
	\begin{align}
	\ev_{Q}={\vvt_k- \vv_Q^{(k-1)}\over \|\vvt_k- \vv_Q^{(k-1)}\|}, \;\; \ev'_Q={\vvt_k- \vv_Q^{(k)}\over \|\vvt_k- \vv_Q^{(k)}\|}.
	\end{align}
	Note that since $\|\ev'_{Q}\|=\|\ev'\|=1$, it follows that
	\begin{align}
	\|\ev'_{Q}-\ev'\|^2&=2(1-\langle \ev'_{Q}, \ev'\rangle).
	\end{align}
	Let 
	\begin{equation}
	\mvec{\zeta}^{(k)}\triangleq \vv^{(k)}-\vv_Q^{(k)}.
	\end{equation}
	By assumption, $g_k$ is $L_k$ Lipschitz. Therefore, $\|\mvec{\zeta}^{(k)}\|\leq L_k\| [\uv^{(k)}]_b-\uv_Q^{(k)}\|\leq L_k 2^{-b}\sqrt{{\eta_k}}$.
	Note that using this definition, 
	\begin{align}
	1-\langle \ev'_{Q}, \ev'\rangle &=1-\langle \frac{\vvt_k- \vv^{(k)}}{\|\vvt_k- \vv^{(k)}\|},{\vvt_k- \vv_Q^{(k)}\over \|\vvt_k- \vv_Q^{(k)}\|}\rangle 
	\nonumber\\
	&=1-\langle \frac{\vvt_k- \vv^{(k)}}{\|\vvt_k- \vv^{(k)}\|},{\vvt_k- \vv^{(k)} +\mvec{\zeta}^{(k)}\over \|\vvt_k- \vv_Q^{(k)}\|}\rangle \nonumber\\
	&=1-{ \|\vvt_k- \vv^{(k)}\| \over \|\vvt_k- \vv_Q^{(k)}\|}-{\langle\mvec{\zeta}^{(k)}, \vvt_k- \vv_Q^{(k)}\rangle\over  \|\vvt_k- \vv_Q^{(k)}\| \|\vvt_k- \vv^{(k)}\| }.
	\end{align}
	Therefore, using Cauchy-Schwartz  and triangle inequalities, it follows that 
	\begin{align}
	1-\langle \ev'_{Q}, \ev'\rangle
	&\leq \Big|1- { \|\vvt_k- \vv^{(k)}\| \over \|\vvt_k- \vv_Q^{(k)}\|}\Big|+ {\|\mvec{\zeta}^{(k)}\|\over \|\vvt_k- \vv_Q^{(k)}\|}.
	\end{align}
	But  $\|\vvt_k- \vv_Q^{(k)}\|-\|\mvec{\zeta}^{(k)}\| \leq  \|\vvt_k- \vv^{(k)}\| \leq  \|\vvt_k- \vv_Q^{(k)}\|+\|\mvec{\zeta}^{(k)}\|$. Therefore,
	\begin{align}
	1-\langle \ev'_{Q}, \ev'\rangle
	&\leq \|\mvec{\zeta}^{(k)}\|\Big( {1\over  \|\vvt_k- \vv^{(k)}\|}+{1\over  \|\vvt_k- \vv_Q^{(k)}\|}\Big),
	\end{align}
	and similarly, $1-\langle \ev_{Q}, \ev\rangle\leq \|\mvec{\zeta}^{(k-1)}\|( {1\over  \|\vvt_k- \vv^{(k-1)}\|}+{1\over  \|\vvt_k- \vv_Q^{(k-1)}\|})$.
	As argued earlier, $\mvec{\zeta}^{(k)}\leq  L_k 2^{-b}\sqrt{{\eta_k}}$ and $\mvec{\zeta}^{(k-1)}\leq  L_k 2^{-b}\sqrt{{\eta_k}}$. Moreover, by assumption, $ \|\vvt_k- \vv^{(k)}\|\geq \sqrt{nB}\delta_k$. Using the triangle inequality, $ \|\vvt_t- \vv_Q^{(k)}\|\geq \sqrt{nB}\delta_k-\| \vv^{(k)}- \vv_Q^{(k)}\|\geq  \sqrt{nB}\delta_k(1-\gamma_k)$, where $\gamma_k$ is defined in \eqref{Eq:gammak}.
	Therefore, $1-\langle \ev'_{Q}, \ev'\rangle\leq\gamma_k(1+{1\over 1-\gamma_k})  $, and 
	\begin{align}
	\|\ev'_{Q}-\ev'\|\leq (\gamma_k(1+{1\over 1-\gamma_k}) )^{1\over 2}.\label{eq:ep-min-epQ-bd-1}
	\end{align}
	Similarly, by assumption, $ \|\vvt_k- \vv^{(k-1)}\|\geq \sqrt{nB}\delta_k$. Therefore, following similar steps, we have
	\begin{align}
	\|\ev_{Q}-\ev\|\leq (\gamma_k(1+{1\over 1-\gamma_k}) )^{1\over 2}.\label{eq:e-min-eQ-bd-1}
	\end{align}
	
	On the other hand, 
	\begin{eqnarray}
	&&\langle \ev,\ev'\rangle-{B}\langle  \Rmat\inv \Hmat \ev, \Hmat \ev'\rangle \nonumber\\
	&=&\langle (\ev-\ev_Q+\ev_Q),\ev'\rangle-{B}\langle  \Rmat\inv \Hmat  (\ev-\ev_Q+\ev_Q), \Hmat \ev'\rangle\\
	&=&\langle \ev-\ev_Q,\ev'\rangle-{B}\langle  \Rmat\inv \Hmat  (\ev-\ev_Q), \Hmat \ev'\rangle
	+\langle \ev_Q,\ev'\rangle-{B}\langle  \Rmat\inv \Hmat  \ev_Q, \Hmat \ev'\rangle \nonumber\\
	&=&\langle \ev-\ev_Q,\ev'\rangle-{B}\langle  \Rmat\inv \Hmat  (\ev-\ev_Q), \Hmat \ev'\rangle
	+\langle \ev_Q,\ev'-\ev'_Q\rangle \nonumber\\
	&&-{B}\langle  \Rmat\inv \Hmat  \ev_Q, \Hmat (\ev'-\ev'_Q)\rangle
	+\langle \ev_Q,\ev'_Q\rangle-{B}\langle  \Rmat\inv \Hmat  \ev_Q, \Hmat \ev'_Q\rangle
	.
	\end{eqnarray}
	Therefore, applying the Cauchy-Schwartz inequality, we have 
	\begin{align}
	&\left|(\langle \ev,\ev'\rangle-{B}\langle  \Rmat\inv \Hmat \ev, \Hmat \ev'\rangle)
	-(\langle \ev_Q,\ev'_Q\rangle-{B}\langle  \Rmat\inv \Hmat \ev_Q, \Hmat \ev'_Q\rangle)\right|\nonumber\\
	&\leq (1+2{B}\sigma_{\max}(\Hmat\ts\Rmat\inv\Hmat))( \|\ev-\ev_Q\|+\|\ev'-\ev'_Q\|).\label{eq:diff-between-e-eQ-1}
	\end{align}
	
	Next we bound $\langle \ev_Q,\ev'_Q\rangle-{B}\langle  \Rmat\inv \Hmat \ev_Q, \Hmat \ev'_Q\rangle$. Recall that $\Rmat={\rm diag}(R_1,\ldots,R_B)$ and $\Hmat=[\Dmat_1,\ldots,\Dmat_B]$, where $\Dmat_j$ is an $n\times n$ diagonal matrix. Moreover, dividing $\ev_Q,\ev'_Q\in\mathds{R}^{nB}$ into $B$ blocks of length $n$, we can write $\ev_Q=[(\ev_{Q,1}\ts,\ldots,(\ev_{Q,B})\ts]\ts$ and $\ev'_Q=[(\ev'_{Q,1})\ts,\ldots,(\ev'_{Q,B})\ts]\ts$. Therefore, using this notation, we have
	\begin{align}
	&\langle \ev_Q,\ev'_Q\rangle-{B}\langle  \Rmat\inv \Hmat \ev_Q, \Hmat \ev'_Q\rangle \nonumber\\
	&=\sum_{i=1}^n\Big(\sum_{b=1}^B e_{Q,bi}e_{Q,bi}^{(1)} 
	-\frac{{B}}{R_i}{\sum_{b_1=1}^B  D_{b_1i}e_{Q,b_1i}}\sum_{b_2=1}^B  D_{b_2i}e'_{Q,b_2i}\Big)\nonumber\\
	&=\sum_{i=1}^nX_i,
	\end{align}
	where for $i=1,\ldots,n$,  random variable $X_i$ is defined as 
	\begin{align}
	X_i\triangleq \sum_{b=1}^B e_{Q,bi}e'_{Q,bi}
	-\frac{{B}}{R_i}{\sum_{b_1=1}^B  D_{b_1i}e_{Q,b_1i}}\sum_{b_2=1}^B  D_{b_2i}e'_{Q,b_2i}. 
	\end{align}
	
	We next prove that $\{X_i\}$ are a sequence of  independent zero-mean and bounded random variables. First, since $D_{bi}$, $(b,i)\in\{1,\ldots,B\}\times\{1,\ldots,n\}$, are i.i.d.$\sim\Nc(0,1)$, it is straightforward to observe that $X_1,\ldots,X_n$ are independent random variables, as each one only depends on a disjoint subset of them. Furthermore, since $D_{b_1i}$ and $D_{b_2i}$ are independent and have symmetric distributions, for $b_1\neq b_2$,  
	\begin{align}
	\Eox{D_{b_1i}D_{b_2i}\over R_i}=0.\label{eq:Eox-D-b1-b2-1}
	\end{align}
	But, 
	\begin{align}
	\Eox{X_i}= \sum_{b=1}^B e_{Q,bi}e'_{Q,bi}
	-{B}{\sum_{b_1=1}^B\sum_{b_2=1}^B  \Eox{D_{b_1i}D_{b_2i}\over R_i}e_{Q,b_1i}}  e'_{Q,b_2i}.
	\end{align}
	Therefore, using \eqref{eq:Eox-D-b1-b2-1}, it follows that 
	\begin{align}
	\Eox{X_i}=\sum_{b=1}^B e_{Q,bi}e'_{Q,bi}-{B}{\sum_{b=1}^B  \Eox{D_{bi}^2\over R_i}e_{Q,b_1i}}  e'_{Q,b_2i}.
	\end{align} 
	On the other hand $ \sum_{b=1}^B\Eox{D_{bi}^2\over R_i}=1$, and given the symmetry of the problem, $\Eox{D_{1i}^2\over R_i}=\ldots=\Eox{D_{Bi}^2\over R_i}$. Therefore, $\Eox{D_{bi}^2\over R_i}={1\over B}$, for all $b$, which shows that $\Eox{X_i}=0$, as claimed. Finally, we show the boundedness of $X_i$. By the Cauchy-Schwartz inequality, 
	\begin{align}
	|X_i-\sum_{b=1}^B e_{Q,bi}e'_{Q,bi}|
	&=\frac{{B}}{R_i}|{\sum_{b_1=1}^B  D_{b_1i}e_{Q,b_1i}}\sum_{b_2=1}^B  D_{b_2i}e'_{Q,b_2i}|\nonumber\\
	&\leq \frac{{B}}{R_i} (\sum_{b=1}^B  D_{bi}^2)(\sum_{b=1}^B (e_{Q,bi})^2)^{1\over 2}(\sum_{b=1}^B (e'_{Q,bi})^2)^{1\over 2}\nonumber\\
	&= {B} (\sum_{b=1}^B (e_{Q,bi})^2)^{1\over 2}(\sum_{b=1}^B (e'_{Q,bi})^2)^{1\over 2}.\label{eq:boundedness-Xi-1}
	\end{align}
	Therefore, given $\ev_{Q}$ and $\ev'_{Q}$, using the Hoeffding's inequality, it follows that
	\begin{align}
	\Pr\Big(\sum_{i=1}^n(\Eox{X_i}-X_i)\geq \lambda \Big)
	\leq 
	\exp\Big(-{2\lambda^2\over  4{B}^2  \sum_{i=1}^n \sum_{b=1}^B (e_{Q,bi})^2\sum_{b=1}^B (e'_{Q,bi})^2}\Big).\label{eq:Hoeffding-main-1}
	\end{align}
	
	Define the set of normalized quantized error vectors as 
	\begin{align}
	\Fc_{Q,k}\triangleq  
	\Big\{
	{\vvt_k-g_k( [\uv]_b)\over \|\vvt_t-g_k( [\uv]_b)\|}: \; \uv\in\Uc^{\eta_k}, \; {1\over \sqrt{nB}}  \|\vvt_k-g_k( \uv])\|\geq \delta_k
	\Big\}
	\end{align}
	Then, given $\lambda>0$, define event ${\Ec}$ as
	\begin{align}
	{\Ec}\triangleq \Big\{&\LPd{\ev_Q,\ev'_Q}- {B}\LPd{ \Hmat \ev_Q, \Rmat\inv \Hmat \ev'_Q} \leq \lambda, \;\forall \;(\ev_Q, \ev'_Q)\in\Fc_{\Qc,k}^2 \Big\}.\label{eq:proof-gap_event_1-1}
	\end{align}
	
	Consider 	
	\[
	\ev_Q={\vvt_k-g_k( [\uv]_b)\over \|\vvt_t-g_k( [\uv]_b)\|},
	\]
	where $ \|\vvt_k-g_k( \uv)\|\geq \delta_k\sqrt{nB}$. By the triangle inequality,  $ \|\vvt_k-g_k( [\uv]_b)\|\geq  \|\vvt_k-g_k(\uv)\| -  \|g_k( [\uv]_b)-g_k(\uv)\|$. But, $\|g_k( [\uv]_b)-g_k(\uv)\|\leq L_k 2^{-b}\sqrt{\eta_k}$. Therefore, combining the two bounds, it follows that
	\[
	\|\vvt_k-g_k( [\uv]_b)\|\geq \delta_k\sqrt{nB}-L_k 2^{-b}\sqrt{\eta_k}.
	\]
	Therefore, 
	\[
	\|\ev_Q\|_{\infty}\leq  {\rho\over  \delta_k\sqrt{nB}-L_k 2^{-b}\sqrt{\eta_k}}={\rho\over  \delta_k\sqrt{nB}(1-\gamma_k)},
	\]
	where $\gamma_k$ is defined in \eqref{Eq:gammak}.
	As a result 
	\begin{align}
	\sum_{i=1}^n \sum_{b=1}^B (e_{Q,bi})^2\sum_{b=1}^B (e'_{Q,bi})^2\leq {\rho^4\over n\delta_k^4(1-\gamma_k)^4}.
	\end{align}
	
	Therefore, \eqref{eq:boundedness-Xi-1} can be further upper bounded as follows:
	\begin{align}
	\Pr\Big(\sum_{i=1}^n(\Eox{X_i}-X_i)\geq \lambda \Big)
	\leq 
	\exp\Big(-{2\lambda^2 n\delta_k^4(1-\gamma_k)^4\over  4 {B}^2\rho^4 }\Big).
	\end{align}
	Hence, by the union bound,
	\begin{align}
	\Pr(\Ec^c)\leq |\Uc_Q|^{2\eta_k}\exp\Big(-{2\lambda^2 n\delta_k^4(1-\gamma_k)^4\over  4 {B}^2\rho^4 }\Big),
	\end{align}
	where $\Uc_{Q}\triangleq \{[\uv]_b:\;\uv\in\Uc\}$.
	
	Finally, we bound $\sigma_{\max}(\Hmat\ts\Rmat\inv \Hmat)$. Consider $\xv\in\mathds{R}^{nB}$ and let $\xv=[\xv_1\ts,\ldots,\xv_B\ts]\ts$, where $\xv_b\in\mathds{R}^n$.  First, note that since $\Hmat=[\Dmat_1,\ldots,\Dmat_B]$,  $\Hmat\xv=\sum_{b=1}^B\Dmat_b\xv_b$ and $\Rmat\inv\Hmat\xv=\sum_{b=1}^B\Rmat\inv\Dmat_b\xv_b$. Let $\qv=\Rmat\inv\Hmat\xv$. Then, $\zv\in\mathds{R}^n$ and 
	\[
	q_i={\sum_{b=1}^B D_{bi}x_{bi} \over \sum_{b=1}^B D^2_{bi} }.
	\]
	Let $\qv'=\Hmat\ts\Rmat\inv \Hmat\xv=\Hmat\ts\qv$. Assume that $\qv'=[(\qv'_1)\ts,\ldots,(\qv'_B)\ts]\ts$, where $\qv'_b\in\mathds{R}^n$. Then,
	\[
	q'_{bi}={D_{bi}\sum_{b'=1}^B D_{bi}x_{b'i} \over \sum_{b'=1}^B D^2_{b'i} }.
	\]
	Therefore, 
	\begin{align}
	\|\Hmat\ts\Rmat\inv \Hmat\xv\|^2=\|\qv'\|^2
	&=\sum_{i=1}^n\sum_{b=1}^B{D^2_{bi}(\sum_{b'=1}^B D_{b'i}x_{b'i})^2 \over (\sum_{b'=1}^B D^2_{b'i})^2 } \nonumber\\
	&=\sum_{i=1}^n{(\sum_{b'=1}^B D_{b'i}x_{b'i})^2 \over \sum_{b'=1}^B D^2_{b'i} }\nonumber\\
	&\stackrel{\rm (a)}{\leq} \sum_{i=1}^n{\sum_{b'=1}^B D^2_{b'i}\sum_{b'=1}^B x^2_{b'i} \over \sum_{b'=1}^B D^2_{b'i} }\nonumber\\
	&= \sum_{i=1}^n\sum_{b'=1}x^2_{b'i}=\|\xv\|^2,
	\end{align}
	where $\rm (a)$ follows from Cauchy-Schwarz inequality. Therefore, $\sigma_{\max}(\Hmat\ts\Rmat\inv \Hmat)\leq 1$.

	Combining \eqref{eq:error-reduction-main-1}, \eqref{eq:diff-between-e-eQ-1} and $\sigma_{\max}(\Hmat\ts\Rmat\inv \Hmat)\leq 1$, conditioned on $\Ec$, it follows that 
	\begin{align*}
	{1\over \sqrt{nB}}\|\vvt_k- \vv^{(k)}\|\leq &2\left| \lambda + (1+2{B})( \|\ev-\ev_Q\|+\|\ev'-\ev'_Q\|) \right| \|\vvt_k- \vv^{(k-1)}\| \nonumber\\
	&+2{B}\|\xv^* -\vvt_k \|\nonumber\\
	\leq &{2\over \sqrt{nB}}( \lambda +\alpha_k) \|\vvt_k- \vv^{(k-1)}\|+2{B}\delta_k,\label{eq:final-reduction-ineq-1}
	\end{align*}
	where the last line follows from \eqref{eq:ep-min-epQ-bd-1} and \eqref{eq:e-min-eQ-bd-1} and $\alpha_k$ is defined in \eqref{Eq:alphak}. Finally, note that by the triangle inequality,
	\begin{align}
	\|\vvt_k- \vv^{(k-1)}\|&\leq  \|\vvt_{k-1}- \vv^{(k-1)}\|+ \|\vvt_k- \vvt_{k-1}\|\nonumber\\
	&\leq  \|\vvt_{k-1}- \vv^{(k-1)}\|+ \|\vvt_k- \xv^*\|+ \|\vvt_{k-1}- \vv\|\nonumber\\
	&\leq  \|\vvt_{k-1}- \vv^{(k-1)}\|+\sqrt{nB}(\delta_k+\delta_{k-1}),
	\end{align}
	where the last line follows from our assumption that $g_{k-1}$ and $g_{k}$ cover $\Qc$ with distortion $\delta_{k-1}$ and $\delta_k$, respectively. Combining this with \eqref{eq:final-reduction-ineq-1} yields the desired result. 
\end{proof}
As mentioned earlier, the main result follows directly from Theorem \ref{The:GAP_SCI_pro} as follows. 
\begin{proof}[Proof of Theorem~\ref{The:GAP_SCI_pro_2}] 
	In Theorem \ref{The:GAP_SCI_pro},  set quantization at iteration $k$ as $b_k=\lceil (1-\zeta) \log {1\over \delta_k }\rceil$. Then, $\gamma_k$  and $\alpha_k$ defined in \eqref{Eq:gammak}  and \eqref{Eq:alphak}, respectively, can be bounded as
	\begin{align}
	\gamma_k\leq  L_k  \delta_k^{\zeta}\sqrt{\eta_k \over nB},
	\end{align}
	Moreover, $|\Uc_Q|\leq 2^{2b}\leq 2^{ (1-\zeta) \log {1\over \delta_k }+1}$, and therefore, 
	\begin{align}
	&\sum_{i=0}^{r-1} |\Uc^{(k)}_Q|^{2\eta_k}\exp(-{2\lambda^2 n\delta_k^4(1-\gamma_k)^4\over  4 {B}^2\rho^4 }) \nonumber\\
	\leq &\exp(-{2\lambda^2 n\delta_k^4(1-\gamma_k)^4\over  4 {B}^2\rho^4 }+ 2\ln 2((1-\zeta) \log {1\over \delta_k }+1)\eta_k).
	\end{align}
\end{proof}

\section{Conclusions}
We proposed GAP-net, a deep unfolding technique for snapshot compressive imaging. GAP-net is a unified framework that leads to state-of-the-art performance on both video and spectral SCI with theoretical convergence guarantees.
It can reconstruct more than 60 data cubes per second for both systems with a spatial size of $256\times 256$.
The demonstrated excellent real data results and fast reconstruction speed suggest that GAP-net can be embedded into real SCI cameras to provide end-to-end real-time capture and reconstruction, and is thus ready to be widely used in  practical applications. 
Currently we are developing energy-efficient networks for SCI cameras to be deployed on robots and self-driving vehicles enabling them to capture rich information.
\bibliographystyle{spmpsci}      


\end{document}